\documentclass[onecolumn]{IEEEtran}
\usepackage{amssymb,amsfonts,amsmath,amsthm}
\usepackage[dvips]{graphicx}
\usepackage{overpic}
\usepackage{subfigure}
\usepackage{color}
\usepackage{cases}
\usepackage{enumerate}
\usepackage{framed}
\definecolor{shadecolor}{rgb}{0.92,0.92,0.92}
\usepackage{listings}
\usepackage{graphicx}
\usepackage{float}
\usepackage{mathrsfs}
\usepackage[top=1.0in, bottom=1.08in, left=0.75in, right=0.75in]{geometry}
\usepackage{pgf}
\usepackage{tikz}
\usetikzlibrary{arrows, automata, positioning, calc, shapes}
\usepackage{comment}
\usepackage[numbers,sort&compress]{natbib}
\usepackage{caption}
\usepackage{hyperref}
\usepackage{bm} 
\usepackage[ruled,linesnumbered]{algorithm2e}
\usepackage{balance}
\usepackage{makecell}

\newtheorem{theorem}{Theorem}

\newtheorem{example}{Example}
\newtheorem{definition}{Definition}
\newtheorem{lemma}{Lemma}
\newtheorem{remark}{Remark}

\newtheorem{proposition}{Proposition}

\newtheorem{corollary}{Corollary}

\newcommand{\bu}{\bm{u}}

\newcommand{\cC}{\mathcal{C}}

\usepackage{comment}

\title{Bounds on Multiple $b$-Burst Deletion-Correcting Codes}\vspace{-1ex}

\author{Chen~Wang, Xiangliang~Kong, Eitan~Yaakobi, and Tolga~M.~Duman%
\thanks{This research was partially funded by the European Union (DiDAX, 101115134) and the ERC Advanced Grant 101054904: TRANCIDS. Views and opinions expressed are however those of the authors only and do not necessarily reflect those of the European Union or the European Research Council Executive Agency. Neither the European Union nor the granting authority can be held responsible for them. This research was partially supported by the Israel Science Foundation (ISF) under Grant No.~2462/24 and by a DFG Middle-East Grant under Grant No.~554629494. The research of C. Wang was supported in part at the Technion by a fellowship from the Lady Davis Foundation.}
\thanks{C. Wang (cwang@campus.technion.ac.il) and E. Yaakobi (yaakobi@cs.technion.ac.il) are with the Department of Computer Science, Technion -- Israel Institute of Technology, Haifa 3200003, Israel. X. Kong (rongxlkong@gmail.com) and Tolga M. Duman (duman@ee.bilkent.edu.tr) are with the Department of Electrical and Electronics Engineering, Bilkent University, Ankara 06800, Turkey.}
\thanks{An earlier version of this paper has been accepted for presentation at the 2026 IEEE International Symposium on Information Theory (ISIT), Guangzhou, China, June 28--July 3, 2026~\cite{ISITversion}.}
}

\begin{document}

\maketitle

\begin{abstract}
	Motivated by their applications in DNA-based storage systems, codes capable of correcting consecutive deletions have attracted significant attention. An important class of such codes consists of those that can correct multiple consecutive deletion errors, commonly referred to as \emph{multiple $b$-burst deletion-correcting codes}. In this paper, we investigate the fundamental limits of multiple $b$-burst deletion-correcting codes. Specifically, we first characterize several structural properties of the associated deletion balls. Then, leveraging these properties, we derive several upper bounds and a combinatorial lower bound on the maximum size of such codes. As a consequence, our bounds improve upon the previously known results for general parameter regimes and are shown to be asymptotically optimal for certain cases.
\end{abstract}

\section{Introduction}

Codes designed for channels with synchronization errors modeled as insertions and deletions have attracted significant attention due to their applications in disk and DNA-based data storage, racetrack memories, file synchronization, and compression \cite{MBT10,YKGMZM15,BLCCSS16,YGM17,HMG19}. 

As a key characteristic of DNA-based storage systems, data stored in DNA molecules is often corrupted by bursts of insertions and deletions \cite{LKGBC19}, whereas substitution errors dominate in traditional optical or magnetic storage systems. Motivated by this observation, many works have focused on designing codes capable of correcting consecutive bursts of deletions and on exploring the fundamental limits of the corresponding code parameters, for examples, see \cite{cheng2014codes,schoeny2017codes,schoeny2017novel,saeki2018improvement,sima2020syndrome,LP20,WTSGF24,nguyen2024new,sun2025codes,sun2025asymptotically,ye2025codes}.

In this paper, motivated by similar considerations, we study the maximal size of codes capable of correcting multiple bursts of deletions. Specifically, by exploring the structural properties of deletion balls corresponding to $t$ different $b$-burst deletions (see Section~\ref{subsec: relevant results} for a formal definition), we derive several upper and lower bounds on the maximum size of such codes. Our analysis combines the linear programming framework of Kulkarni and Kiyavash \cite{kulkarni2013nonasymptotic}, sphere-packing arguments, and hypergraph matching theory. Our general upper bound recovers the existing results for the special cases $t = 1$ and $b = 1$, and improves upon the best-known bounds in the general setting. Our lower bound establishes the asymptotic tightness of the upper bound when $q$ is sufficiently large.

In the remainder of this section, we first briefly review related prior work and then present our results.

\subsection{Previous Work and Relevant Results}\label{subsec: relevant results}

For an integer $q \ge 2$, let $\Sigma_q = \{0,1,\ldots,q-1\}$ denote the $q$-ary alphabet. 
For positive integers $t$, $b$, and $n$ satisfying $n \ge tb + 1$, and for a sequence $\mathbf{x} \in \Sigma_q^{n}$, we denote by $D_t^b(\mathbf{x})$ (respectively, $I_t^b(\mathbf{x})$) the set of all subsequences of $\mathbf{x}$ obtained by applying $t$ consecutive $b$-burst deletions (respectively, insertions), i.e., deleting (respectively, inserting) exactly $t$ contiguous blocks of $b$ symbols. We refer to $D_t^b(\mathbf{x})$ and $I_t^b(\mathbf{x})$ as the \emph{$(t,b)$-burst-deletion ball} and \emph{$(t,b)$-burst-insertion ball} centered at $\mathbf{x}$, respectively. For simplicity, we omit the parameter $b$ from the notation $D_t^b(\mathbf{x})$ and $I_t^b(\mathbf{x})$ when $b=1$.

A code $\mathcal{C} \subseteq \Sigma_q^{n}$ is said to be a \emph{$(t,b)$-burst-deletion-correcting code} if for any two distinct codewords $\mathbf{c}, \mathbf{c}' \in \mathcal{C}$, $D_t^b(\mathbf{c}) \cap D_t^b(\mathbf{c}') = \varnothing$. We denote by $M_q(n,(t,b))$ the largest size of a $(t,b)$-burst-deletion-correcting code in $\Sigma_q^{n}$.

The study of bounds on $M_q(n,(t,b))$ was initiated by Levenshtein~\cite{Levenshtein66}, who showed that $|D_1^b(\mathbf{x})| = |U_b(\mathbf{x})|$, where, for any $\mathbf{x} \in \Sigma_q^n$, the set $U_b(\mathbf{x})$ is defined  as 
$$
U_b(\mathbf{x}) \triangleq \{i \in [n-b] : x_i \neq x_{i+b}\} \cup \{n-b+1\},
$$
and $[n-b] \triangleq \{1,2,\ldots,n-b\}$. When $b=1$, the quantity $\lvert U_b(\mathbf{x}) \rvert$ coincides with the number of runs in $\mathbf{x}$. Subsequently, based on the above characterization of $\left|D_1^b(\mathbf{x})\right|$, Schoeny \emph{et al.} \cite{schoeny2017codes} derived an upper bound on the maximum size of binary $(1,b)$-burst-deletion-correcting codes. This result was later extended to general alphabets by Wang \emph{et al.} \cite{WTSGF24}; see Table~\ref{tab:upper_bounds} for explicit expressions of these bounds.

In a recent work \cite{LSYG26}, Lan \emph{et al.} studied the sequence reconstruction problem over a channel subject to multiple bursts of insertions and deletions. Specifically, they proved that the size of the $(t,b)$-burst-insertion ball is independent of its center and provided an exact expression for it. Based on this result, they further derived a general upper bound on $M_q(n,(t,b))$; see Table~\ref{tab:upper_bounds} for the resulting expression.

Beyond the above upper bounds, several constructions for $(t,b)$-burst deletion codes have been proposed. Levenshtein \cite{levenshtein1967asymptotically} constructed binary $(1,2)$-burst-deletion codes with redundancy at most $\log (n) + 1$. This was subsequently generalized by Cheng \emph{et al.}~\cite{cheng2014codes} to $(1,b)$-burst-deletion codes with redundancy $b \log (n/b + 1)$, and further improved by Schoeny \emph{et al.}~\cite{schoeny2017codes} to $\log n + (b-1)\log \log n + O(1)$. More recently, Sun \emph{et al.}~\cite{sun2025asymptotically} showed that for $q$-ary alphabets with $q \ge 2$, $(1,b)$-burst-deletion codes can achieve redundancy $\log n + O(1)$. In contrast, for $t \ge 2$, only a limited number of constructions are known. A general framework based on syndrome compression \cite{sima2020syndrome} combined with suitable pre-coding yields $(t,b)$-burst-deletion codes with redundancy at most $(4t - 1)\log n + o(\log n)$ for all $q \ge 2$ and constant $b$. For the special case $t=2$ and constant $b$, the best known construction achieves a redundancy of $5 \log n + o(\log n)$ \cite{ye2025codes}.

\subsection{Our Results}

Our contributions in this paper are as follows:
\begin{itemize}
    \item First, we study structural properties of the $(t,b)$-burst-deletion ball $D_t^b(\mathbf{x})$ for a general sequence $\mathbf{x}\in\Sigma_q^n$. In particular, we derive upper and lower bounds on $|D_t^b(\mathbf{x})|$ and establish a monotonicity property: for any sequence $\mathbf{x}\in D_1^b(\mathbf{z})$, it holds that $|D_t^b(\mathbf{x})| \le |D_t^b(\mathbf{z})|$ (see Lemma~\ref{lem: monotone property of |D_t^b(x)|}).

    \item Second, based on these structural properties of $(t,b)$-burst-deletion balls, we derive both a linear-programming bound and a sphere-packing bound on $M_q(n,(t,b))$, which asymptotically behave as
    $$
    M_q(n,(t,b)) \le \frac{t!\, q^{n - tb + t}}{(q - 1)^t \left(n - 2tb - \frac{(t-1)b}{q}\right)^t}\left(1 + o_n(1)\right),
    $$
    as $n \to \infty$ for fixed $q$ and $t$. Moreover, building on a connection between $(t,b)$-burst-deletion-correcting codes and matchings in a special hypergraph, we show that
    $$
    M_q(n,(t,b))=
    \begin{cases}
    q^{b}, & \text{if } t = \frac{n}{b} - 1,\\
    \frac{q^{n - tb}}{\binom{n - tb + t}{t}}(1 - o(1)), & \text{if } t < \frac{n}{b} - 1,
    \end{cases}
    $$
    as $q \to \infty$ for fixed $n, t, b$, which recovers the result for $b=1$ in \cite{kong2025combinatorial}.
    
    \item Finally, we further investigate the behavior of $|D_t^b(\mathbf{x})|$ and its dependence on the center sequence $\mathbf{x}$ for the special case $t = b = 2$. By analyzing the relationship between subsequences in $D_2^2(\mathbf{x})$ and those in $D_2(\mathbf{x})$, we derive an improved lower bound on $|D_2^2(\mathbf{x})|$, which is shown to be tight in certain cases. Consequently, this yields an improved upper bound on $M_q(n,(2,2))$ compared to the upper bounds obtained for general parameter settings.
\end{itemize}

In Table~\ref{tab:upper_bounds}, we summarize the upper bounds on $M_q(n,(t,b))$ obtained in this paper and compare them with several known results. As shown in the table, our general upper bound asymptotically coincides with the bounds in~\cite{schoeny2017codes,WTSGF24} for the case $t=1$, and with those in~\cite{Levenshtein66,kulkarni2013nonasymptotic} for the case $b=1$. Moreover, it improves upon the bound in~\cite{LSYG26} for the general parameter setting where $b,t \ge 2$ and $q \ge 2$. Furthermore, we also list the corresponding redundancies, which yield non-existence results for $(t,b)$-burst deletion-correcting codes with such redundancies.

\begin{table}[h!]
    \centering
    \renewcommand{\arraystretch}{1.8}
    \caption{Summary of upper bounds on $M_q(n, (t, b))$.}
    \label{tab:upper_bounds}
    \begin{tabular}{cccc}
    \hline
    Parameter Regimes & Upper bounds & Redundancy & Reference\\\hline
    $q = 2, b = 1$ & $\frac{2^n t!}{n^t}$ & $t\log{n}-\log{t!}$ & \cite{Levenshtein66}\\\hline
    $b = 1$ & $\frac{t!q^n}{(q-1)^tn^t}$ & $t\log_q{n}-\log_q(\frac{t!}{(q-1)^{t}})$ & \cite{kulkarni2013nonasymptotic}\\\hline
    $t = 1$ & $\frac{q^{n - b + 1} - q^b}{(q - 1)(n - 2b + 1)}$ & $b+\log_q(n-2b)-\log_q(\frac{q}{q-1})+o(1)$ & \thead{\cite{WTSGF24}\\ (\cite{schoeny2017codes} for $q=2$)}\\ \hline
    \thead{any $t, b\geq 1$,\\ $q\geq 2$} & $\frac{q^{n + t}}{\sum_{i = 0}^t\binom{n + t}{i}(q - 1)^i}$ & $t\log{n}-t\log_q(\frac{q}{q-1})-\log{t!}+o(1)$ & \cite{LSYG26}\\\hline
    \thead{any $b\geq 1$,\\ $q\geq 2$, $t\geq 1$ constants} & $\frac{t!q^{n - tb + t}(1 + o(1))}{(q - 1)^t\left(n-2tb-\frac{(t-1)b}{q}\right)^t}
    $ & $tb+t\log{n}-t\log_q(\frac{q}{q-1})-\log{t!}+O(1)$ & Theorems \ref{thm1: upper bound on M_q(n,(t,b))} and \ref{thm: upper_bound_levenshtein}\\
    \hline
    \thead{$t=b=2$,\\any $q\geq2$} & $\frac{2q^{n-4}}{(q-1)^2n^2}(1+o(1))$ & $2\log_q{n}+4-\log_q(\frac{2}{(q-1)^2})+o(1)$ & Corollary \ref{coro: asy up for t=b=2}\\
    \hline
    \end{tabular}
\end{table}
% \begin{table}[h!]
%     \centering
%     \renewcommand{\arraystretch}{1.8}
%     \caption{Summary of upper bounds on $M_q(n, (t, b))$.}
%     \label{tab:upper_bounds}
%     \begin{tabular}{cccc}
%     \hline
%     Upper bounds & Redundancy & Parameter Regimes & Reference\\\hline
%     $\frac{2^n t!}{n^t}$ & $t\log{n}-\log{t!}$ & $q = 2, b = 1$ & \cite{Levenshtein66}\\\hline
%     $\frac{t!q^n}{(q-1)^tn^t}$ & $t\log_q{n}-\log_q(\frac{t!}{(q-1)^{t}})$ & $b = 1$ & \cite{kulkarni2013nonasymptotic}\\\hline
%     %$\frac{2^{n - b + 1} - 2^b}{n - 2b + 1}$ & $n-b+1-\log(n-2b)+o(1)$ & $q = 2, t = 1$ &\cite{schoeny2017codes}\\\hline
%     $\frac{q^{n - b + 1} - q^b}{(q - 1)(n - 2b + 1)}$ & $b+\log_q(n-2b)-\log_q(\frac{q}{q-1})+o(1)$ & $t = 1$ & \thead{\cite{WTSGF24}\\ (\cite{schoeny2017codes} for $q=2$)}\\ \hline
%     $\frac{q^{n + t}}{\sum_{i = 0}^t\binom{n + t}{i}(q - 1)^i}$ & $t\log{n}-t\log_q(\frac{q}{q-1})-\log{t!}+o(1)$ & \thead{any $t, b\geq 1$,\\ $q\geq 2$} & \cite{LSYG26}\\\hline
%     $\frac{t!q^{n - tb + t}(1 + o(1))}{(q - 1)^t\left(n-2tb-\frac{(t-1)b}{q}\right)^t}
%     $ & $tb+t\log{n}-t\log_q(\frac{q}{q-1})-\log{t!}+O(1)$ & \thead{any $b\geq 1$,\\ $q\geq 2$, $t\geq 1$ constants} & Theorems \ref{thm1: upper bound on M_q(n,(t,b))} and \ref{thm: upper_bound_levenshtein}\\
%     \hline
%     $\frac{2q^{n-4}}{(q-1)^2n^2}(1+o(1))$ & $2\log_q{n}+4-\log_q(\frac{2}{(q-1)^2})+o(1)$ & \thead{$t=b=2$,\\any $q\geq2$} & Corollary \ref{coro: asy up for t=b=2}\\
%     \hline
%     \end{tabular}
% \end{table}
%\XLcom{Don't forget to extend Table \ref{tab:upper_bounds} to include all the results.}

The rest of the paper is organized as follows. In Section~\ref{sec: notations}, we introduce the necessary notation. In Section~\ref{sec: properties of D_t^b(x)}, we present the first part of our results on the structural properties of the $(t,b)$-burst-deletion ball $D_t^b(\mathbf{x})$ for general $t$ and $b$. In Section~\ref{sec: general bounds}, we establish two upper bounds on $M_q(n,(t,b))$ for general parameter regimes, along with combinatorial upper and lower bounds. In Section~\ref{sec: improved ub t=b=2}, we present our results for $D_2^2(\mathbf{x})$ and the proof of the improved upper bound on $M_q(n,(t,b))$. Finally, we conclude the paper in Section~\ref{sec: conclusion}.

\section{Notations}\label{sec: notations}

For integers $m$ and $n$, the set $\{m, m+1, \ldots, n\}$ is denoted by $[m:n]$, and $[n]$ is used as shorthand for $[1:n]$.
For a sequence $\mathbf{x} = (x_1, \ldots, x_n) \in \Sigma_q^n$ and a subset $R \subseteq [n]$, let $\mathbf{x}_R$ denote the restriction of $\mathbf{x}$ to the coordinates indexed by $R$.
In particular, we define $\mathbf{x}_{\emptyset}$ to be the null vector, i.e., the vector of length zero.
For two sequences $\mathbf{x}$ and $\mathbf{y}$, we denote by $\mathbf{x}\mathbf{y}$ their concatenation.

We also define the $b \times \lceil n / b \rceil$ \emph{array representation} of $\mathbf{x}$, denoted by $A_b(\mathbf{x})$, obtained by arranging the entries of $\mathbf{x}$ column by column. If the last column is incomplete, the remaining entries are filled by repeating the last symbol in the corresponding row. Specifically,
$$
A_b(\mathbf{x}) \triangleq
\begin{bmatrix}
    x_1   & x_{b+1} & \cdots & x_{n-b+1} \\
    x_2   & x_{b+2} & \cdots & x_{n-b+2} \\
    \vdots& \vdots  & \ddots & \vdots    \\
    x_b   & x_{2b}  & \cdots & x_n
\end{bmatrix}.
$$
Moreover, we also use
$$
    \begin{bmatrix}
    \mathbf{x}^{(1)}; 
    \mathbf{x}^{(2)};
    \cdots;
    \mathbf{x}^{(\lceil n / b \rceil)}
    \end{bmatrix}
$$
to denote the row-wise form of $A_{b}(\mathbf{x})$, where $\mathbf{x}^{(j)} \triangleq (x^{(j)}_1, x^{(j)}_2, \ldots, x^{(j)}_{\lceil n / b \rceil})$ and $x^{(j)}_i = x_{j+(i-1)b}$ for all $j\in [b]$ and $i\in [\lceil n / b \rceil]$.

For a positive integer $b$, a sequence $\mathbf{y} \in \Sigma_q^{\,n-b}$ is said to be obtained from $\mathbf{x} \in \Sigma_q^n$ by a \emph{$b$-burst deletion} if 
$$
\mathbf{y} = \mathbf{x}_{[n] \setminus [i:i+b-1]}
$$ 
for some $i \in [n-b+1]$. More generally, for integers $i_1, \ldots, i_t \in [n-b+1]$ satisfying $i_{j+1} - i_j \ge b$ for all $j \in [t-1]$, define
\begin{align}\label{eq:subseq_in_D_t^b}
    \mathrm{Del}_b(\mathbf{x}, \{i_1, \ldots, i_t\})
    \triangleq
    \mathbf{x}_{[n] \setminus \bigcup_{j=1}^t [i_j : i_j + b - 1 ]}
\end{align}
as the subsequence obtained by applying $t$ disjoint $b$-burst deletions at positions $i_1, \ldots, i_t$. Then, the $(t,b)$-burst-deletion ball $D_t^b(\mathbf{x})$ is defined as\vspace{-1.2ex}
\begin{align}
    D_t^b(\mathbf{x})
    \triangleq
    \Bigl\{
    \mathrm{Del}_b(\mathbf{x}, \{i_1, \ldots, i_t\}) :
    ~ i_j \in [n-b+1], 
    i_{j+1}-i_j \ge b,\ \forall\, j \in [t-1]
    \Bigr\}. \label{eq: def D_t^b(x)}
\end{align}
Given $\mathbf{y} \in D_t^b(\mathbf{x})$, we say that $\mathrm{Del}_b(\mathbf{x}, \{i_1, \ldots, i_t\})$ is a \emph{representation} of $\mathbf{y}$ if
$\mathbf{y} = \mathrm{Del}_b(\mathbf{x}, \{i_1, \ldots, i_t\})$.

%Let $d_b(\mathbf{x}, \mathbf{y}) = n/b - t + 1$ be the distance of $\mathbf{x}, \mathbf{y}$, where $t$ is the minimum integer such that $D_t^b(\mathbf{x})\cap D_t^b(\mathbf{y})\neq \varnothing$. 
%\XLcom{Is the definition of this distance $d_b(\cdot,\cdot)$ standard? Because it doesn't seem to include edit distance as a special case.}

Similarly, for a positive integer $b$, we denote by $I_t^b(\mathbf{x})$ the set of all supersequences of $\mathbf{x}$ obtained after $t$ $b$-burst insertions, which we refer to as the \emph{$(t,b)$-burst-insertion ball} centered at $\mathbf{x}$. A code $\mathcal{C} \subseteq \Sigma_q^{n}$ is said to be a \emph{$(t,b)$-burst-insertion-correcting code} if for any two distinct codewords $\mathbf{c}, \mathbf{c}' \in \mathcal{C}$, $I_t^b(\mathbf{c}) \cap I_t^b(\mathbf{c}') = \varnothing$. Employing a similar approach with \cite{Levenshtein66}, it can be verified that a code is a $(t,b)$-burst-deletion-correcting code if and only if it is a $(t,b)$-burst-insertion-correcting code. Therefore, we restrict our discussion to $(t,b)$-burst-deletion-correcting codes in this work.

\section{Properties of $(t,b)$-burst-deletion balls and bounds on their sizes}\label{sec: properties of D_t^b(x)}

In this section, we first establish several properties of representations of sequences in $D_t^b(\mathbf{x})$. Using these properties, we then derive upper and lower bounds on $|D_t^b(\mathbf{x})|$ for any $\mathbf{x} \in \Sigma_q^n$, and prove a monotonicity property of $|D_t^b(\mathbf{x})|$.

To begin with, we introduce a compact representation for sequences in $D_t^b(\mathbf{x})$, which groups together consecutive $b$-burst deletions occurring at adjacent positions.

Consider a sequence $\mathrm{Del}_b(\mathbf{x}, \{j_1, j_2, \ldots, j_t\})\in D_t^b(\mathbf{x})$ satisfying
$j_{\ell+1} - j_\ell \ge b$ for all $\ell \in [t-1]$.
Whenever several deletions occur at positions of the form
$i, i+b, \ldots, i+(c-1)b$ for some $c \ge 1$, we group them into a single
block and represent them by the pair $(i, c)$. Consequently, the deletion pattern $(j_1, j_2, \ldots, j_t)$ can be
represented by an integer $s \in [t]$ with tuples
$(i_1, i_2, \ldots, i_s) \in [n-b+1]^s$ and
$(c_1, c_2, \ldots, c_s) \in \mathbb{Z}_+^s$ satisfying\vspace{-1.2ex}
\begin{align}
    i_{\ell+1} > i_\ell + c_\ell b,~\forall\, 1 \le \ell \le s-1,~\text{and}~\sum_{\ell=1}^s c_\ell = t,\label{eq: representation cond}
\end{align}
as well as $\{j_1, \ldots, j_t\} = \bigcup_{\ell=1}^s \{i_\ell, i_\ell + b, \ldots, i_\ell + (c_\ell - 1) b\}$. Then, we denote this sequence by $\mathrm{Del}_b(\mathbf{x}, (i_1, c_1), \ldots, (i_s, c_s))$ and refer to this expression as a \emph{compact representation} of
$\mathrm{Del}_b(\mathbf{x}, \{j_1, j_2, \ldots, j_t\})$.

\begin{example}\label{ex:compact_representation}
Let $q = 3$, $b = 2$, and $t = 2$. For
$\mathbf{x} \triangleq (0,2,0,1,1,1,2,0) \in \Sigma_3^8$, $\mathbf{y}\triangleq (0,1,2,0)\in D_2^2(\mathbf{x})$ has the following four different representations
\begin{align*}
   \mathrm{Del}_2(\mathbf{x}, \{1, 4\}),\mathrm{Del}_2(\mathbf{x}, \{1, 5\}),\mathrm{Del}_2(\mathbf{x}, \{2, 4\}),\mathrm{Del}_2(\mathbf{x}, \{2, 5\}).
\end{align*}
In the first expression, the $2$-deletions occur at positions $1$ and $4$ are separate, which yields the compact representation $\mathrm{Del}_2(\mathbf{x}, (1,1), (4,1)).$ In the third expression, the $2$-deletions are consecutive, yielding the compact representation $\mathrm{Del}_2(\mathbf{x}, (2,2)).$
\end{example}

Example \ref{ex:compact_representation} shows that a sequence $\mathbf{y} \in D_t^b(\mathbf{x})$ may admit different compact representations. Suppose $\mathbf{y} \in D_t^b(\mathbf{x})$ admits two compact representations 
$$
\mathrm{Del}_b(\mathbf{x}, \hspace{-0.15ex}(i_1, c_1),\hspace{-0.15ex}\dots, \hspace{-0.15ex}(i_s, c_s))\hspace{-0.2ex},~\hspace{-0.2ex}\mathrm{Del}_b(\mathbf{x}, \hspace{-0.15ex}(i'_1, c'_1),\hspace{-0.15ex}\dots, \hspace{-0.1ex}(i'_s, c'_{s'}))\hspace{-0.15ex}.
$$ 
If $(i_\ell, c_\ell) = (i'_\ell, c'_\ell)$ for $\ell < a$, and either
$i_{a} < i'_{a}$, or $i_{a} = i'_{a}$ and $c_{a} > c'_{a}$, we say that $\mathrm{Del}_b(\mathbf{x}, (i'_1, c'_1), \dots, (i'_{s'}, c'_{s'}))$ \emph{precedes} $\mathrm{Del}_b(\mathbf{x}, (i_1, c_1), \dots, (i_s, c_s))$. For example, in Example \ref{ex:compact_representation}, $\mathrm{Del}_2(\mathbf{x}, (2,2))$ precedes $\mathrm{Del}_2(\mathbf{x}, (1,1), (4,1))$.

% \begin{example}
% Given $\mathbf{x} = (0,2,0,1,1,1,2,0) \in \Sigma_3^8$ as in Example \ref{ex:compact_representation}. Then, we have $\mathrm{Del}_2(\mathbf{x}, (2,2))$ precedes $\mathrm{Del}_2(\mathbf{x}, (1,1), (4,1))$.
% \end{example}

In the following, we select a compact representation of $\mathbf{y}$, thereby eliminating ambiguity and ensuring that each sequence admits a unique representation.

\begin{definition}\label{def: maximal representation}
    For $\mathbf{x}\in \Sigma_{q}^{n}$ and $\mathbf{y}\in D_{t}^{b}(\mathbf{x})$, a compact representation $\mathrm{Del}_b(\mathbf{x}, (i_1, c_1), \dots, (i_s, c_s))$ of $\mathbf{y}$ is called \textbf{maximal}, if no other compact representation of $\mathbf{y}$ precedes it.
\end{definition}

\begin{example}
    Given $\mathbf{x} = (0,2,0,1,1,1,2,0) \in \Sigma_3^8$ and  $\mathbf{y} = (0,1,2,0) \in D_2^2(\mathbf{x})$ as in Example \ref{ex:compact_representation}, the sequence $\mathbf{y}$ admits four different compact representations, $\mathrm{Del}_2(\mathbf{x}, (1,1), (4,1))$, $\mathrm{Del}_2(\mathbf{x}, (1,1), (5,1))$, $\mathrm{Del}_2(\mathbf{x}, (2,2)),$ and $\mathrm{Del}_2(\mathbf{x}, (2,1), (5,1))$. Clearly, $\mathrm{Del}_2(\mathbf{x}, (2,1), (5,1))$ is the maximal compact representation of $\mathbf{y}$.
\end{example}

It can be verified that the relation \emph{precedes} defines a total order on the set of compact representations of any given sequence in $D_t^b(\mathbf{x})$. As a consequence, every sequence in $D_t^b(\mathbf{x})$ admits a unique maximal compact representation, as stated in the following proposition.

\begin{proposition}\label{prop: max rep of seq in D_t^b(x) is unique}
	For any $\mathbf{x} \in \Sigma_q^n$, let $\mathbf{u}$ and $\mathbf{v}$ be two sequences in $D_t^b(\mathbf{x})$ with maximal compact representations 
    \[
    \mathrm{Del}_b(\mathbf{x}, (i_1, c_1), \dots, (i_s, c_s)),\mathrm{Del}_b(\mathbf{x}, (i'_1, c'_1), \dots, (i'_{s'}, c'_{s'})), 
    \] 
    respectively. Then, $\mathbf{u} = \mathbf{v}$ if and only if $s = s'$ and $(i_\ell, c_\ell) = (i'_\ell, c'_\ell)$ for any $\ell \in [s].$
\end{proposition}

Moreover, the following proposition provides an equivalent characterization of maximal compact representations.

\begin{proposition}\label{prop: every seq in D_t^b(x) has maximal rep}
Let $\mathbf{x} \in \Sigma_q^n$ and $\mathbf{y} \in D_t^b(\mathbf{x})$, and let $\mathrm{Del}_b(\mathbf{x}, (i_1, c_1), \dots, (i_s, c_s))$ be a compact representation of $\mathbf{y}$.  
Then, $\mathrm{Del}_b(\mathbf{x}, (i_1, c_1), \dots, (i_s, c_s))$ is maximal if and only if, for every $a \in [s]$, 
\begin{align}\label{eq: maximal cond}
    x_{i_a + c_a b} \neq x_{i_a + r b},~ \forall~ r \in [0: c_a-1].
\end{align} 
The condition is vacuous whenever $i_a + c_a b > n.$
\end{proposition}

\begin{IEEEproof}
    We first prove the necessity by contradiction. Suppose that $\mathrm{Del}_b(\mathbf{x}, (i_1, c_1), \dots, (i_s, c_s))$ is not maximal, and let $\mathrm{Del}_b(\mathbf{x}, (i_1', c_1'), \dots, (i_{s'}', c_{s'}'))$ be a compact representation of $\mathbf{y}$, which precedes $\mathrm{Del}_b(\mathbf{x}, (i_1, c_1), \dots, (i_s, c_s))$.

    Without loss of generality, let $a \in [s]$ be the smallest index such that either $i_a < i'_a$ or $i_a = i'_a$ and $c_a > c'_a$. Note that for all $\ell < a$, we have $(i_\ell, c_\ell) = (i'_\ell, c'_\ell)$. Let $d \triangleq \sum_{\ell < a} c_\ell b$ denote the total number of symbols deleted before position $i_a$. We next consider the following two cases.

    If $i_a < i'_a$, then $y_{i_a - d} = x_{i_a}$ since $\mathbf{y}$ admits the representation $\mathrm{Del}_b(\mathbf{x}, (i_1', c_1'), \dots, (i_{s'}', c_{s'}'))$. On the other hand, from the representation $\mathrm{Del}_b(\mathbf{x}, (i_1, c_1), \dots, (i_s, c_s))$, we have $y_{i_a - d} = x_{i_a + c_a b} \neq x_{i_a}$, which is a contradiction.

    Similarly, if $i_a = i'_a$ and $c_a > c'_a$, we have $y_{i_a - d} = x_{i_a + c_a' b}$ by the representation $\mathrm{Del}_b(\mathbf{x}, (i_1', c_1'), \dots, (i_s', c_s'))$ and $y_{i_a - d} = x_{i_a + c_a b}\neq x_{i_a + c'_a b}$ by the representation $\mathrm{Del}_b(\mathbf{x}, (i_1, c_1), \dots, (i_s, c_s))$, which is also a contradiction.

    Next, we prove the sufficiency by contradiction. Suppose that there exists $a \in [s]$ and some $r \in [0: c_a-1]$ such that $x_{i_a + r b} = x_{i_a + c_a b}$. We next construct a new compact representation of $\mathbf{y}$ which precedes $\mathrm{Del}_b(\mathbf{x},(i_1, c_1), \dots, (i_s, c_s))$. 

    When $i_{a+1} > i_a + c_a b + 1$ and $r\geq1$, let $s'= s+1$ and $\{(i'_\ell, c'_\ell)\}_{\ell \in [s']}$ be as follows:
    \begin{align}
        (i'_\ell, c'_\ell) =
        \begin{cases}
            (i_\ell, c_\ell), & 1 \le \ell \le a-1, \\
            (i_a, r), & \ell = a, \\
            (i_a + r b + 1,\, c_a - r), & \ell = a+1, \\
            (i_{\ell-1}, c_{\ell-1}), & a+2 \le \ell \le s' .
            \end{cases}
            \label{eq1: prop new compact rep}
    \end{align}
    Clearly, pairs $\{(i_{\ell}',c_{\ell}')\}_{\ell\in [s']}$ satisfy \eqref{eq: representation cond}, which implies that $\mathrm{Del}_{b}(\mathbf{x}, (i'_1, c'_1),\ldots, (i'_{s'}, c'_{s'}))$ is a compact representation of some $\mathbf{y}'\in D_{t}^{b}(\mathbf{x})$. Moreover, the deleted positions of $\mathbf{y}'$ differ from those of $\mathbf{y}$ only on two intervals:
    \begin{align}
        [i_a: i_a + r b - 1 ]
        ~\text{and}~
        [i_a + r b + 1: i_a + c_a b].
        \label{eq1: new compact representation}
    \end{align}
    Recall that in $\mathbf{y}$, the $a$-th deleted interval is
    $[i_a: i_a + c_a b - 1]$. Hence, $\mathbf{y}$ and $\mathbf{y}'$ can be written as
    \begin{align}
        \mathbf{y} & = (\ldots, x_{i_a-1},\, x_{i_a + c_a b},\, x_{i_a + c_a b + 1}, \ldots), \nonumber \\
        \mathbf{y}' & = (\ldots, x_{i_a-1},\, x_{i_a + rb},\, x_{i_a + c_a b + 1}, \ldots). \label{eq1: prop1 y and y'}
    \end{align}
    Since $x_{i_a + r b} = x_{i_a + c_a b}$, \eqref{eq1: prop1 y and y'} implies that $\mathbf{y} = \mathbf{y}'$.

    When $i_{a+1} > i_a + c_a b + 1$ and $r=0$, let $s'= s$ and $\{(i_{\ell}',c_{\ell}')\}_{\ell\in [s']}$ be as follows:
    \begin{align}
        (i_{\ell}',c_{\ell}')=\begin{cases}
            (i_{\ell},c_{\ell}), & 1\leq \ell\leq a-1,\\
            (i_{a}+1,c_a), & \ell=a,\\
            (i_{\ell},c_{\ell}), &a+1\leq \ell\leq s'. \label{eq2: new compact representation}
        \end{cases}
    \end{align}
    Similarly, we can verify that $\mathrm{Del}_{b}(\mathbf{x}, (i'_1, c'_1),\ldots, (i'_{s'}, c'_{s'}))$ is a compact representation of $\mathbf{y}$.

    When $i_{a+1}=i_a+c_ab+1$, we define the pairs $\{(i_{\ell}',c_{\ell}')\}_{\ell\in [s']}$ in the same way as in \eqref{eq1: new compact representation} and \eqref{eq2: new compact representation} for the cases $r\geq 1$ and $r=0$, respectively. The difference is that, for the case $r\geq 1$, we define the new representation as
    $$
    \mathrm{Del}_{b}(\mathbf{x}, (i'_1, c'_1),\ldots,(i_{a}',c_{a}'), (i_{a+1}',c_{a+1}'+c_{a+2}'),\ldots, (i'_{s'}, c'_{s'})),
    $$
    and for the case $r=0$, we define the new representation as
    $$
    \mathrm{Del}_{b}(\mathbf{x}, (i'_1, c'_1),\ldots,(i_{a-1}',c_{a-1}'),(i_{a}',c_{a}'+c_{a+1}'), \ldots,(i'_{s'}, c'_{s'})).
    $$
    Then, since $i_{a+1}=i_a+c_ab+1$, we have $i_{a+1}'+c_a'b=i_{a+2}$ when $r\geq 1$ and $i_{a}'+c_a'b=i_{a+1}$ when $r=0$. Thus, by arguments similar to those above, the newly defined representation is a valid representation of $\mathbf{y}$.
    
    Finally, we conclude that there exists a new compact representation of $\mathbf{y}$ that precedes $\mathrm{Del}_b(\mathbf{x}, (i_1, c_1), \dots, (i_s, c_s))$, contradicting the assumption that $\mathrm{Del}_b(\mathbf{x}, (i_1, c_1), \dots, (i_s, c_s))$ is maximal. This completes the proof.
\end{IEEEproof}

\begin{example}
    Given $\mathbf{x}$ and  $\mathbf{y}$ as in Example \ref{ex:compact_representation}, consider the compact representation $\mathrm{Del}_2(\mathbf{x}, (2, 2))$ of $\mathbf{y}$. Since $x_4 = x_6$, we can shift the deleted intervals to obtain the maximal compact representation $\mathrm{Del}_2(\mathbf{x}, (2, 1), (5, 1))$, which precedes $\mathrm{Del}_2(\mathbf{x}, (2, 2))$.
\end{example}

Next, based on the above characterization of maximal compact representations of sequences in $D_{t}^{b}(\mathbf{x})$, we prove the following upper and lower bounds on $|D_{t}^{b}(\mathbf{x})|$. This result can be viewed as a generalization of the result by Levenshtein~\cite{Levenshtein66}.
%in the spirit of \cite{Levenshtein66}.

\begin{theorem}\label{thm: size of burst-deletion ball}
    For any $\mathbf{x}\in\Sigma_q^n$, it holds that \begin{align}
        \binom{|U_b(\mathbf{x})|-(t - 1)b}{t}\leq |D_t^b(\mathbf{x})|\leq \binom{|U_b(\mathbf{x})|+t-1}{t}. \nonumber
    \end{align}  
\end{theorem}

\begin{IEEEproof}
    By Propositions~\ref{prop: max rep of seq in D_t^b(x) is unique} and~\ref{prop: every seq in D_t^b(x) has maximal rep}, it suffices to show that, for any $\mathbf{x} \in \Sigma_q^n$, the number of sets of pairs $\{(i_j, c_j)\}_{j \in [s]}$ satisfying \eqref{eq: representation cond} and \eqref{eq: maximal cond} satisfies the stated bounds.

    Note that, by the definition of $U_b(\mathbf{x})$, the following collection of pair sets
    \begin{align}
        \left\{\hspace{-0.25ex}\{(i_j,1)\}_{j\in [t]}\hspace{-0.25ex}:\hspace{-0.25ex}i_j\hspace{-0.25ex}\in \hspace{-0.25ex}U_b(\mathbf{x}),i_{j+1}\hspace{-0.25ex}> \hspace{-0.25ex}i_j + b,\forall~j\in [t-1]\hspace{-0.25ex}\right\} \label{eq: pf Thm bounds1}
    \end{align}
    % \begin{align}
    %     \Bigl\{ \{(i_j,1)\}_{j\in [t]} \; :\;& i_j\in U_b(\mathbf{x}), \nonumber\\
    %     & i_{j+1}> i_j + b,\ \forall\, j\in [t-1] \Bigr\} \label{eq: pf Thm bounds1}
    % \end{align}
    satisfies the conditions \eqref{eq: representation cond} and \eqref{eq: maximal cond}. Hence, the lower bound follows directly, since there are at least $\binom{|U_b(\mathbf{x})|-(t-1)b}{t}$ different choices of $(i_1,i_2,\ldots,i_t)$ satisfying \eqref{eq: pf Thm bounds1}.

    Next, we proceed to prove the upper bound. Observe that each set of pairs $\{(i_j, c_j)\}_{j \in [s]}$ satisfying \eqref{eq: representation cond} and \eqref{eq: maximal cond} is uniquely associated with a multiset of indices of the form
    \begin{align}
        \{\underbrace{i_1,\ldots,i_1}_{c_1}, \underbrace{i_2,\ldots,i_2}_{c_2}, \ldots, \underbrace{i_s,\ldots,i_s}_{c_s}\}, \label{eq: pf Thm bounds2}
    \end{align}
    where $i_j\in U_b(\mathbf{x})$, $i_{j+1} > i_j + c_j b$, and $\sum_{j\in [s]} c_j = t$.
    Hence, the number of sets of pairs $\{(i_j, c_j)\}_{j \in [s]}$ satisfying \eqref{eq: representation cond} and \eqref{eq: maximal cond} is upper bounded by the number of multisets of the form \eqref{eq: pf Thm bounds2}, with $(i_1,\ldots,i_s)\in U_b(\mathbf{x})^{s}$ and $(c_1,\ldots,c_s)\in \mathbb{Z}_{+}^{s}$ satisfying
    $$
    i_j \neq i_{j'},\ \text{for all } j\neq j', \quad \text{and} \quad \sum_{j\in [s]} c_j = t. $$
    The number of such multisets is given by $\binom{|U_b(\mathbf{x})|+t-1}{t}$, which concludes the proof.
\end{IEEEproof}

Next, we establish the following property of $|D_{t}^{b}(\mathbf{x})|$.

%, which states that the size of the deletion ball decreases when $\mathbf{x}$ is replaced by a subsequence of $\mathbf{x}$ in $D_{1}^{b}(\mathbf{x})$. Specifically, we have the following lemma.

\begin{lemma}\label{lem: monotone property of |D_t^b(x)|}
    For any sequences $\mathbf{z} \in \Sigma_q^{n+b}$ and $\mathbf{x} \in D_1^b(\mathbf{z}) \subseteq \Sigma_q^{n}$, it holds that $|D_t^b(\mathbf{x})| \le |D_t^b(\mathbf{z})|.$
\end{lemma}

\begin{IEEEproof}
    Let $\mathbf{x} = \mathbf{u}\mathbf{v}$ and $\mathbf{z} = \mathbf{u}\bm{\sigma}\mathbf{v}$ for some $\mathbf{u}, \mathbf{v} \in \Sigma_q^*$ and $\bm{\sigma} \in \Sigma_q^b$. We aim to construct an injective map $\phi$ from $D_t^b(\mathbf{x})$ to $D_t^b(\mathbf{z})$. The result then follows directly.

    By Propositions~\ref{prop: max rep of seq in D_t^b(x) is unique} and~\ref{prop: every seq in D_t^b(x) has maximal rep}, every sequence $\mathbf{y} \in D_t^b(\mathbf{x})$ has a unique maximal compact representation $\mathrm{Del}_b(\mathbf{x}, (i_1, c_1),$ $ (i_2, c_2), \ldots, (i_s, c_s))$, and the same property holds for every sequence in $D_t^b(\mathbf{z})$. Next, we construct a map $\phi$ from $D_t^b(\mathbf{x})$ to $D_t^b(\mathbf{z})$ that preserves the maximality of compact representations. The injectivity of $\phi$ then follows directly from the uniqueness of the maximal compact representation. The construction of $\phi$ proceeds in two separate cases.

    \textbf{Case 1}: For all $1\leq \ell\leq s$, either $i_{\ell} + c_{\ell}b\leq |\mathbf{u}|$ or $i_{\ell}>|\mathbf{u}|$.
    
    In Case 1, we define $\phi(\mathbf{y})\triangleq \mathrm{Del}_b(\mathbf{z}, (i_1', c_1), (i_2', c_2), $ $\dots, (i_s', c_s))$, where 
    \begin{align}
        i_\ell' = \begin{cases}
		i_\ell, & \text{if}~ i_\ell \leq |\bu|-c_{\ell}b, \\
		i_\ell + b, & \text{if}~i_\ell > |\bu|.
	\end{cases} \label{eq: phi case1}
    \end{align}
    Clearly, pairs $\{(i'_{\ell}, c_{\ell})\}_{\ell \in [s]}$ satisfy \eqref{eq: representation cond} and for each $\ell\in [s]$, it also holds that either $i'_{\ell} + c'_{\ell}b \leq  |\bu|$ or $i'_{\ell} > |\bu|$. Moreover, since $\mathrm{Del}_b(\mathbf{x}, (i_1, c_1), (i_2, c_2), \ldots, (i_s, c_s))$ is maximal, it holds that $x_{i_{\ell} + c_{\ell} b} \neq x_{i_{\ell} + r b}$, for all $r \in [0: c_{\ell} - 1]$ and $\ell \in [s]$.  Thus, by writing $\mathbf{z} = \mathbf{u}\bm{\sigma}\mathbf{v}$, we have
    \begin{align}
        z_i =
        \begin{cases}
        x_i, & 1 \le i \le |\mathbf{u}|, \\
        \sigma_{i - |\mathbf{u}|}, & |\mathbf{u}| + 1 \le i \le |\mathbf{u}| + b, \\
        x_{i - b}, & |\mathbf{u}| + b + 1 \le i \le n + b,
        \end{cases} \label{eq: expression z_i}
    \end{align}
    which, combined with \eqref{eq: phi case1}, implies that
    $$
    z_{i'_{\ell} + c_{\ell} b} \neq z_{i'_{\ell} + r b}, ~\forall\, r \in [0: c_{\ell} - 1],~\forall\, \ell \in [s].
    $$
    %{\color{red}Furthermore, in this case, $[i_{\ell}', i_{\ell}' + c_{\ell}'b]$ is disjoint from $[|\mathbf{u}| + 1, |\mathbf{u}| + b]$ for all $\ell \in [s]$.}

    \textbf{Case 2}: There exists some $\ell_0 \in [s]$ such that $i_{\ell_0} \le |\mathbf{u}| < i_{\ell_0} + c_{\ell_0} b.$
    
    Let $r_0\in [c_{\ell_0}]$ be the minimum integer such that $|\mathbf{u}|<i_{\ell_0}+r_0b\leq |\mathbf{u}|+b$. Let $\lambda\in [b]$ be such that $|\mathbf{u}| + \lambda = i_{\ell_0}+r_0b$. Next, we consider the following two subcases.

	\textbf{Case 2.1}: $x_{i_{\ell_0} + c_{\ell_0}b} \ne \sigma_{\lambda}$.
    
    In Case 2.1, we define $\phi(\mathbf{y})\triangleq \mathrm{Del}_b(\mathbf{z}, (i_1', c_1), (i_2', c_2), $ $\dots, (i_s', c_s))$, where 
    \begin{align}
        i_\ell' = \begin{cases}
		i_\ell, & \text{if } i_\ell \le |\mathbf{u}|~\text{and}~\ell < \ell_0, \\
		i_\ell + b, & \text{if }~\ell =\ell_0~\text{or}~i_\ell > |\mathbf{u}|.
	\end{cases} \label{eq: phi case2.1}
    \end{align}
    Clearly, pairs $\{(i'_{\ell}, c_{\ell})\}_{\ell \in [s]}$ satisfy \eqref{eq: representation cond}. 
    %Note that, by \eqref{eq: phi case2.1}, we have $i'_{\ell} = i_{\ell} < |\mathbf{u}| - c_{\ell} b$ for all $\ell < \ell_0$, and $i'_{\ell} = i_{\ell} + b > |\mathbf{u}|$ for all $\ell > \ell_0$.
    Then, by \eqref{eq: expression z_i} and the maximality of $\mathrm{Del}_b(\mathbf{x}, (i_1, c_1), (i_2, c_2), \ldots, (i_s, c_s))$, we have $z_{i'_{\ell} + c_{\ell} b} \neq z_{i'_{\ell} + r b}$, for all $r \in [0 : c_{\ell} - 1], \ \ell \in [s] \setminus \{\ell_0\}$. Moreover, by the definitions of $r_0$ and $\lambda$, \eqref{eq: expression z_i} implies that
    \begin{align}
    z_{i'_{\ell_0} + r b}= z_{i_{\ell_0} + r b + b}
    &=\begin{cases}
        x_{i_{\ell_0} + (r+1)b}, & \text{if } 0 \le r < r_0 - 1, \\
        \sigma_{\lambda}, & \text{if } r = r_0 - 1, \\
        x_{i_{\ell_0} + r b}, & \text{if } r_0 \le r \le c_{\ell_0}.
       \end{cases}
    \nonumber
    \end{align}
    Thus, since $x_{i_{\ell_0} + c_{\ell_0} b} \neq \sigma_{\lambda}$, we obtain $z_{i'_{\ell_0} + c_{\ell_0} b} \neq z_{i'_{\ell_0} + r b}$, for all $r \in [0 : c_{\ell_0} - 1]$.
    %{\color{red}Furthermore, the pair $(i_{\ell_0}', c_{\ell_0})$ satisfies $i_{\ell_0}' \leq |\mathbf{u}| + b < i_{\ell_0}' + c_{\ell_0} b$, and the interval $[i_\ell : i_\ell + c_\ell b]$ is disjoint from $[|\mathbf{u}| + 1: |\mathbf{u}| + b]$ for all $\ell \in [s] \setminus \{\ell_0\}$. Hence, this case is different from the previous ones.}

    \textbf{Case 2.2}: $x_{i_{\ell_0} + c_{\ell_0}b} = \sigma_{\lambda}$.
    
    In Case 2.2, we define 
    \begin{align}
        \phi(\mathbf{y})\triangleq & \mathrm{Del}_b\left(\mathbf{z}, (i_1', c_1), \ldots, (i'_{\ell_0-1},c_{\ell_0-1}), (i'_{\ell_0},c'_{\ell_0}), \right. \nonumber \\
        & ~~~~~~~~~~\left.(i''_{\ell_0},c''_{\ell_0}),(i'_{\ell_0+1},c_{\ell_0+1}),\dots, (i_s', c_s)\right) \nonumber
    \end{align}
	where $(i_{\ell_0}', c_{\ell_0}')=(i_{\ell_0}, r_0)$ and $(i_{\ell_0}'', c_{\ell_0}'')=(i_{\ell_0}+(r_0+1)b, c_{\ell_0} - r_{0})$ and\vspace{-0.5ex}
	$$
    i_\ell' = \begin{cases}
		i_\ell, & \text{if } i_\ell \le |\mathbf{u}|~\text{and}~\ell < \ell_0, \\
		i_\ell + b, & \text{if } i_\ell > |\mathbf{u}|.
	\end{cases}\vspace{-0.5ex}
    $$
	Clearly, pairs $\{(i'_{\ell}, c_{\ell})\}_{\ell \in [s] \setminus \{\ell_0\}}\cup \{(i'_{\ell_0}, c'_{\ell_0}), (i''_{\ell_0}, c''_{\ell_0})\}$ satisfies~\eqref{eq: representation cond}. 
    
    By \eqref{eq: expression z_i} and the maximality of $\mathrm{Del}_b(\mathbf{x}, (i_1, c_1), (i_2, c_2), $ $ \ldots, (i_s, c_s))$, we have $z_{i'_{\ell} + c_{\ell} b} \neq z_{i'_{\ell} + r b}$, for all $r \in [0 : c_{\ell} - 1],\ \ell \in [s] \setminus \{\ell_0\}$. Moreover, 
    by the definitions of $r_0$ and $\lambda$, \eqref{eq: expression z_i} implies that\vspace{-0.5ex}
    \begin{align}
    z_{i'_{\ell_0} + r b}
    &=\begin{cases}
        x_{i_{\ell_0} + rb}, & \text{if } 0 \le r < r_0, \\
        \sigma_{\lambda}, & \text{if } r = r_0.
       \end{cases} \nonumber
    \end{align}
    and $z_{i''_{\ell_0} + r b}=x_{i_{\ell_0} + (r+r_0)b}$ for all $r\in [0: c_{\ell_0}-r_0]$. Since $x_{i_{\ell_0} + c_{\ell_0}b} = \sigma_{\lambda}$, we have $z_{i'_{\ell_0} + c'_{\ell_0} b}=x_{i_{\ell_0} + c_{\ell_0}b}\neq z_{i'_{\ell_0} + r b}$ for all $r\in [0: c'_{\ell_0}-1]$ and $z_{i''_{\ell_0} + c''_{\ell_0} b}=x_{i_{\ell_0} + c_{\ell_0}b}\neq z_{i''_{\ell_0} + r b}$ for all $r\in [0: c''_{\ell_0}-1]$ by $x_{i_{\ell_0}+c_{\ell_0}b}\neq x_{i_{\ell_0}+rb}$ for all $r\in [0: c_{\ell_0}-1]$. 

    Note that for the pairs $\{(i'_\ell, c'_\ell)\}_{\ell \in [s]}$ in the representation $\phi(\mathbf{y})$ in Case 1, it holds that either $i'_\ell + c'_\ell b \leq |\mathbf{u}|$ or $i'_\ell > |\mathbf{u}|$ for all $\ell \in [s]$. In contrast, in Case 2, there always exists a pair $(i'_\ell, c'_\ell)$ in $\phi(\mathbf{y})$ such that $i'_\ell \leq |\mathbf{u}| \leq i'_\ell + c'_\ell b$. This shows that the image sets under $\phi$ of sequences with maximal representations in Case 1 and Case 2 are disjoint. Moreover, the number of pairs in $\phi(\mathbf{y})$ satisfies $s' = s$ in Case 2.1, while $s' = s + 1$ in Case 2.2. Therefore, the image sets under $\phi$ of sequences corresponding to Cases 2.1 and 2.2 are also disjoint.
    %{\color{red} Furthermore, in this case, we have $i'_{\ell_0} \leq |\mathbf{u}| + 1 \leq i'_{\ell_0} + c'_{\ell_0}b < |\mathbf{u}| + b$, $i_{\ell_0}'' > |\mathbf{u}| + b$, and $[i_{\ell}',i_{\ell}' + c_{\ell}'b]$ is disjoint from $[|\mathbf{u}| + 1, |\mathbf{u}| + b]$ for all $\ell \in [s]\backslash\{\ell_0\}$, which is different from the previous cases.}

    To sum up, for both cases, we show that the defined map $\phi$ preserves the maximality of the compact representation of the corresponding sequence. This completes the proof.
    %the above construction and discussions yield a map $\phi$ from $D_{t}^{b}(\mathbf{x})$ to $D_{t}^{b}(\mathbf{z})$ which preserves the maximality of the compact representation of the corresponding sequence. Thus, the result follows by Proposition \ref{prop: every seq in D_t^b(x) has maximal rep} and \ref{prop: max rep of seq in D_t^b(x) is unique}.
\end{IEEEproof}

\begin{example}
    Given $\mathbf{x} = (0, 2, 0, 1, 1, 1, 2, 0)\in \Sigma_3^8$ and $\mathbf{z} = (0, 2, 0, 1, 1, 1, 2, 1, 2, 0)\in \Sigma_3^{10}$, we have $\mathbf{x}\in D_1^2(\mathbf{z})$. Let $\mathbf{u} = (0, 2, 0 ,1 ,1 ,1)$, $\mathbf{\sigma} = (2, 1)$, and $\mathbf{v} = (2, 0)$, we have $\mathbf{x} = \mathbf{u}\mathbf{v}$ and $\mathbf{z} = \mathbf{u}\bm{\sigma}\mathbf{v}$.
    
    Let $t = 2$, $\mathbf{y}_1 = \mathrm{Del}_2(\mathbf{x}, (2, 1), (7, 1)) = (0, 1, 1, 1)$, $\mathbf{y}_2 = \mathrm{Del}_2(\mathbf{x}, (4, 2)) = (0, 2, 0, 0)$, and $\mathbf{y}_3 = \mathrm{Del}_2(\mathbf{x}, (3, 2)) = (0, 2, 2, 0)$. Denote $\phi$ as the map defined in the proof of Lemma \ref{lem: monotone property of |D_t^b(x)|}, the following holds:
    \begin{itemize}
        \item As the deleted positions are chosen to satisfy either $i_{\ell}+c_{\ell}b\leq 6$ or $i_{\ell}>6$, we have $\phi(\mathbf{y}_1) = \mathrm{Del}_2(\mathbf{z}, (2, 1), (9, 1)) = (0, 1, 1, 1, 2, 1)$ (Case 1);
        \item As $x_8 \ne \sigma_2$, we have $\phi(\mathbf{y}_2) = \mathrm{Del}_2(\mathbf{z}, (6, 2)) = (0, 2, 0, 1, 1, 0)$ (Case 2.1);
        \item As $x_7 = \sigma_1$, we have $\phi(\mathbf{y}_3) = \mathrm{Del}_2(\mathbf{z}, (3, 2)) = (0, 2, 2, 1, 2, 0)$ (Case 2.2).
    \end{itemize}
\end{example}

\section{Bounds on $M_{q}(n,(t,b))$ for Arbitrary $t$ and $b$}\label{sec: general bounds}

In this section, we derive three upper bounds on the maximal size of a $(t,b)$-burst-deletion-correcting code, $M_q(n,(t,b))$. The first bound is obtained via a linear program, in the spirit of the bound on $M_q(n,(t,1))$ by Kulkarni and Kiyavash~\cite{kulkarni2013nonasymptotic}. The second is a sphere-packing bound, analogous to that in~\cite{levenshtein2002bounds}. The third is a combinatorial bound, which improves upon the first two when either $q$ and $n-tb$ are fixed with $n$ large, or $n$ is fixed with $q$ large.

\subsection{A Linear Programming Upper Bound}

In this subsection, we present the proof of our first upper bound on $M_q(n,(t,b))$. 
\begin{theorem}\label{thm1: upper bound on M_q(n,(t,b))}
    Let $q \ge 2$ and $t$ be fixed positive integers. For any positive integers $n$ and $b$ satisfying $n \ge 2tb + (t-1)(b+1)$ and $b \mid n$, the maximum size of a $(t,b)$-burst-deletion-correcting code of length $n$ satisfies
    $$
    M_q(n,(t,b)) \le
    \frac{t!\, q^{n - tb + t}}{(q - 1)^t \left(n - 2tb - \frac{(t-1)b}{q}\right)^t}\left(1 + o_n(1)\right).
    $$
\end{theorem}

As in \cite{kulkarni2013nonasymptotic}, we model the problem of finding the largest $(t,b)$-burst-deletion-correcting code as a matching problem on a hypergraph. We then prove Theorem~\ref{thm1: upper bound on M_q(n,(t,b))} by constructing a feasible solution to the dual linear program corresponding to this hypergraph matching formulation.

We start with recalling some necessary notions in the theory of hypergraphs. A \emph{hypergraph} $\mathcal{H}$ is defined as a pair $(V(\mathcal{H}), E(\mathcal{H}))$, where the vertex set $V(\mathcal{H})$ is a finite set and the hyperedge set $E(\mathcal{H})$ is a collection of subsets of $V(\mathcal{H})$. A \emph{matching} in $\mathcal{H}$ is a set of hyperedges that do not share any common vertex. The matching number of $\mathcal{H}$, denoted by $\nu(\mathcal{H})$, is the size of a largest matching in $\mathcal{H}$. 

Consider the hypergraph $\mathcal{H}_{q,n,t}^{b}$ with vertex set $\Sigma_q^{n-tb}$ and hyperedge set $\{ D_t^b(\mathbf{x}) : \mathbf{x} \in \Sigma_q^{n} \}$. That is, each vertex of $\mathcal{H}_{q,n,t}^{b}$ is a sequence of length $n-tb$ over $\Sigma_q$, and a collection of vertices form a hyperedge if and only if they coincide with $D_t^b(\mathbf{x})$ for some $\mathbf{x} \in \Sigma_q^{\,n}$. Then, a $(t,b)$-burst-deletion correcting code in $\Sigma^n$ corresponds to a matching in $\mathcal{H}_{q,n,t}^{b}$ and hence we have $M_q(n,(t,b))=\nu(\mathcal{H}_{q,n,t}^{b})$, where $\nu(\mathcal{H}_{q,n,t}^{b})$ denote the matching number of $\mathcal{H}_{q,n,t}^{b}$. In other words, we have 
\begin{align}
    M_q(n,(t,b))=\textrm{\rm maximize}\quad & \sum_{\mathbf{x}\in \Sigma_q^{n}}z(\mathbf{x}) \nonumber\\
    \textrm{subject to} \quad \sum_{\mathbf{x}\in I_t^{b}(\mathbf{y})} & z(\mathbf{x}) \leq 1,~\forall~\mathbf{y}\in \Sigma_q^{n-tb} \label{eq: cond1 of LP}; \\
     & z(\mathbf{x})\in \mathbb{Z}^+,~\forall~ \mathbf{x}\in \Sigma_q^{n}.\nonumber
\end{align}
Since the feasible regions of the integer programs are strictly contained in the feasible regions of their linear programming relaxations, by the Duality Theorem of linear programming (see \cite[Corollary 7.1g]{Schrijver98}), $M_q(n,(t,b))$ is upper bounded by
\begin{align}
    \textrm{minimize}\quad & \sum_{\mathbf{y}\in \Sigma_q^{n-tb}}w(\mathbf{y}) \nonumber\\
        \textrm{subject to} \quad \sum_{\mathbf{y}\in D_t^{b}(\mathbf{x})} & w(\mathbf{y}) \geq 1,~\forall~\mathbf{x}\in \Sigma_q^{n} \label{eq: cond2 of LP}; \\
        & w(\mathbf{y})\geq 0,~\forall~ \mathbf{y}\in \Sigma_q^{n-tb}.\nonumber
\end{align}

%using the lower bound on $D_{t}^{b}(\mathbf{x})$ in Theorem \ref{thm: size of burst-deletion ball} and Lemma \ref{lem: monotone property of |D_t^b(x)|},

\begin{theorem}\label{thm2: upper bound on M_q(n,(t,b))}
    Let $q,n,t,b$ be positive integers satisfying $q\geq 2$ and $n\geq tb+1$. Then, it holds that
    $$
    M_q(n,(t,b))\leq \sum_{\mathbf{y}\in \Sigma_{q}^{n-tb}}|D_{t}^{b}(\mathbf{y})|^{-1}.
    $$
\end{theorem}

\begin{IEEEproof}
    It suffices to show that $w(\mathbf{y})=|D_{t}^{b}(\mathbf{y})|^{-1}$, for any $\mathbf{y}\in \Sigma_q^{n-tb}$, is a feasible solution for the dual LP problem \eqref{eq: cond2 of LP}. 
    
    Clearly, $w(\mathbf{y})\geq 0$. Moreover, for any $\mathbf{x}\in \Sigma_q^{n}$, 
    \begin{align}
        \sum_{\mathbf{y}\in D_{t}^{b}(\mathbf{x})}w(\mathbf{y}) & = \sum_{\mathbf{y}\in D_{t}^{b}(\mathbf{x})}|D_{t}^{b}(\mathbf{y})|^{-1} \nonumber \\
        & \geq \sum_{\mathbf{y}\in D_{t}^{b}(\mathbf{x})}|D_{t}^{b}(\mathbf{x})|^{-1}=1, \nonumber
    \end{align}
    where the second inequality follows since $|D_{t}^{b}(\mathbf{y})|\leq |D_{t}^{b}(\mathbf{x})|$ by Lemma \ref{lem: monotone property of |D_t^b(x)|}.
\end{IEEEproof}

By substituting the lower bound on $|D_t^b(x)|$ in Theorem \ref{thm: size of burst-deletion ball}, we have the following non-asymptotic upper bound on $M_q(n,(t,b))$.

\begin{theorem}\label{thm: non asymptotic upper bound 1}
    Let $q \ge 2$ and $t$ be fixed positive integers. For any positive integers $n$ and $b$ satisfying $n \ge 2tb + (t-1)(b+1)$ and $b \mid n$, the maximum size of a $(t,b)$-burst-deletion-correcting code of length $n$ satisfies
    $$
    M_q(n,(t,b)) \le q^{b} \sum_{i=r+1}^{n-(t+1)b+1} (q-1)^{i-1} \frac{\binom{n-tb-b}{i-1}}{\binom{i-(t-1)b}{t}} + q^b\sum_{i=1}^{r} (q-1)^{i-1} \binom{n-tb-b}{i-1}
    $$
    for any integer $n-(t+1)b\geq r \ge (t-1)(b+1)$.
\end{theorem}

\begin{IEEEproof}
    By Theorem~\ref{thm: size of burst-deletion ball}, it suffices to upper bound the following summation:
    $$
    \sum_{\mathbf{y} \in \Sigma_q^{\,n-tb}} \lvert D_t^b(\mathbf{y}) \rvert^{-1}.
    $$

    Note that
    \begin{align*}
    \sum_{\mathbf{y}\in \Sigma_q^{\,n-tb}} \lvert D_t^b(\mathbf{y}) \rvert^{-1}
    = \sum_{i=1}^{n-(t+1)b+1}
    \sum_{\substack{\mathbf{y}\in \Sigma_q^{n-tb}, \\ |U_b(\mathbf{y})| = i}}
    \lvert D_t^b(\mathbf{y}) \rvert^{-1}.
    \end{align*}
    Thus, by the lower bound on $\lvert D_t^b(\mathbf{y}) \rvert$ in Theorem~\ref{thm: size of burst-deletion ball}, we have
    \begin{align*}
    \sum_{\mathbf{y}\in \Sigma_q^{n-tb}} \lvert D_t^b(\mathbf{y}) \rvert^{-1} 
    \le & \sum_{i = r+1}^{n-(t+1)b+1}
    \sum_{\substack{\mathbf{y}\in \Sigma_q^{n-tb}, \atop |U_b(\mathbf{y})| = i}}
    \binom{i-(t-1)b}{t}^{-1}  + \sum_{i =1}^{r}
    \sum_{\substack{\mathbf{y}\in \Sigma_q^{n-tb}, \atop |U_b(\mathbf{y})| = i}}
    1 \\
    = & \sum_{i=r+1}^{n-(t+1)b+1}
    \frac{\bigl|\{\mathbf{y}\in \Sigma_q^{n-tb} : |U_b(\mathbf{y})| = i\}\bigr|}
    {\binom{i-(t-1)b}{t}} + \sum_{i =1}^{r} \bigl|\{\mathbf{y}\in \Sigma_q^{n-tb} : |U_b(\mathbf{y})| = i\}\bigr| \\
    = & q^{b} \sum_{i=r+1}^{n-(t+1)b+1} (q-1)^{i-1} \frac{\binom{n-tb-b}{i-1}}{\binom{i-(t-1)b}{t}} + q^b\sum_{i=1}^{r} (q-1)^{i-1} \binom{n-tb-b}{i-1},
    \end{align*}
    for any integer $n-(t+1)b\geq r \ge (t-1)(b+1)$, and the last equality follows from
    \begin{align*}
    \bigl|\{\mathbf{y}\in \Sigma_q^{n-tb} : |U_b(\mathbf{y})| = r\}\bigr|
    = \binom{n-tb-b}{r-1} q^{b} (q-1)^{r-1}.
    \end{align*}
    See Claim~4 in Appendix~A of~\cite{WTSGF24} for a detailed proof.
\end{IEEEproof}

Next, we derive the asymptotic upper bound on $M_q(n,(t,b))$ in Theorem~\ref{thm1: upper bound on M_q(n,(t,b))} by choosing an appropriate value of $r$ in Theorem~\ref{thm: non asymptotic upper bound 1}.

\begin{IEEEproof}[Proof of Theorem \ref{thm1: upper bound on M_q(n,(t,b))}]
    By Theorem~\ref{thm: non asymptotic upper bound 1}, it suffices to estimate the following two terms:
    \begin{align}
        q^{b} \sum_{i=r+1}^{n-(t+1)b+1} (q-1)^{i-1} \frac{\binom{n-tb-b}{i-1}}{\binom{i-(t-1)b}{t}} \label{eq: estimate sum D_t^b(y) term1}
    \end{align}
    and
    \begin{align}
        q^b\sum_{i=1}^{r} (q-1)^{i-1} \binom{n-tb-b}{i-1}.\label{eq: estimate sum D_t^b(y) term2}
    \end{align}

    To further estimate \eqref{eq: estimate sum D_t^b(y) term1} and \eqref{eq: estimate sum D_t^b(y) term2}, let $N=n-(t+1)b$ and $\mathrm{Vol}_q(r, N) = \sum_{i = 0}^r (q - 1)^i\binom{N}{i}$. By the well-known bounds on $\mathrm{Vol}_q(r, N)$ (See \cite[Proposition 3.3.3.]{guruswami2012essential} for details), we have
    \begin{align}\label{eq: vol of Hamming Ball}
    \mathrm{Vol}_q(r, N) \leq q^{H_q(p)N},~ \text{for } p = \frac{r}{N} \leq 1 - \frac{1}{q},
    \end{align}
    where $H_q(p) \triangleq p\log_q(q - 1) - p\log_q p - (1 - p)\log_q(1 - p)$ is the $q$-ary entropy function. It can be verified that
    \begin{align*}
    \frac{\partial H_q(p)}{p} &= \log_q(q-1) - \log_q p + \log_q(1-p),\\
    \frac{\partial^2 H_q(p)}{p^2} &= -\frac{1}{p(1-p)\ln q},
    \end{align*} 
    which implies that the Taylor expansion of $H_q(p - \epsilon)$ at $p = 1 - \frac{1}{q}$ is
    \begin{align}
    H_q(p - \epsilon) = 1 - \frac{q^2}{2(q-1)\ln q} \epsilon^2 + o(\epsilon^2).\label{eq: Taylor expansion of H_q(p-epsilon)}
    \end{align} 

    The term in \eqref{eq: estimate sum D_t^b(y) term2} can be simplified as
    \begin{align}
    q^b \sum_{r=0}^{r'-1}
    (q-1)^{r}
    \binom{N}{r} = q^b\mathrm{Vol}_q(r'-1, N). \label{eq: estimate2 sum D_t^b(y) term2}
    \end{align}
    Moreover, since $\binom{r-(t-1)b}{t}$ is monotone increasing in $r$, thus the term in \eqref{eq: estimate sum D_t^b(y) term1} can be simplified as
    \begin{align}
	& q^b \sum_{r = r'+1}^{n-(t+1)b+1}
    (q-1)^{r-1}
    \frac{\binom{n-tb-b}{r-1}}{\binom{r-(t-1)b}{t}} \nonumber \\
    \leq & \frac{q^b}{\binom{r'+1-(t-1)b}{t}} \sum_{r = r'+1}^{n-(t+1)b+1}
    (q-1)^{r-1}
    \binom{n-tb-b}{r-1} \nonumber \\
    = & \frac{q^{b}}{\binom{r'+1-(t-1)b}{t}} \left(\mathrm{Vol}_q(N, N) - \mathrm{Vol}_q(r' - 1, N) \right).
    \label{eq: estimate2 sum D_t^b(y) term1}
    \end{align}

    Now set $r' = \Big(1 - \frac{1}{q}\Big)N - \sqrt{2tN\ln N} + 1$ in \eqref{eq: estimate2 sum D_t^b(y) term1} and \eqref{eq: estimate2 sum D_t^b(y) term2}. By \eqref{eq: vol of Hamming Ball}, the right-hand side of \eqref{eq: estimate2 sum D_t^b(y) term2} satisfies
    \begin{align*}
    q^b \mathrm{Vol}_q(r'-1, N)
    & \le q^{b + H_q\!\left(1 - \frac{1}{q} - \sqrt{\frac{2t\ln N}{N}}\right)N} \\
    & = q^{N + b - \frac{tq^2}{(q-1)\ln q}\ln N  + o(t\ln N)} \\
    & \le q^{n - tb - tq \log_q N - t\ln N\left(\frac{1}{\ln q} - o(1)\right)} \\
    & = O\!\left(\frac{q^{n-tb}}{(n-(t+1)b)^{qt}}\right),
    \end{align*}
    as $n \to \infty$, where the first inequality follows from \eqref{eq: Taylor expansion of H_q(p-epsilon)} and the second follows since $q^2/(q-1) \ge q+1$ and $\ln N / \ln q = \log_q N$. Meanwhile, by
    \begin{align}
    & {{r'+1}-(t-1)b \choose t} %\nonumber \\
    \geq  \frac{((1 - \frac{1}{q})N-(t-1)b)^{t}}{t!}(1-o(1)) \nonumber \\
    = & \left(\frac{q-1}{q}\right)^{t}\left(n-2tb-\frac{(t-1)b}{q}\right)^{t}\left(\frac{1}{t!}-o(1)\right)
    \end{align}
    and
    \begin{align*}
        \mathrm{Vol}_q(N, N) - \mathrm{Vol}_q(r' - 1, N) & \geq q^{N}-q^{N - \frac{tq^2}{(q-1)\ln q}\ln N  + o(t\ln N)} \\
        & = q^{n-(t+1)b}\left(1-o(1)\right) \\
    \end{align*}
    as $n\rightarrow \infty$ and $t$ is a fixed constant, the RHS of \eqref{eq: estimate2 sum D_t^b(y) term1} is at most
    \begin{align}
    \frac{t! q^{n-tb+t}}{(q-1)^{t}\left(n-2tb-\frac{(t-1)b}{q}\right)^{t}}(1+o(1)). \nonumber
    \end{align}
    In total, this leads to
    \begin{align*}
	\sum_{\mathbf{y} \in \Sigma_q^{n - tb}} |D_{t}^b(\mathbf{y})|^{-1}
    \leq & \frac{t!q^{n - tb + t}}{(q - 1)^t\left(n-2tb-\frac{(t-1)b}{q}\right)^t}(1 + o(1))
    \end{align*}
    and confirms the upper bound in Theorem \ref{thm1: upper bound on M_q(n,(t,b))}.
\end{IEEEproof}

\begin{remark}\label{rmk: LP_bound_comparison_with_known_bounds}
    \begin{itemize}
        \item[1.] When $t = 1$ and $q = 2$, the bound in Theorem~\ref{thm1: upper bound on M_q(n,(t,b))} reduces to
        $$
        M_2(n,(1,b)) \le \frac{2^{n-b+1}}{n-2b}(1 + o(1)),
        $$
        which recovers the bound of Schoeny \emph{et al.}~\cite[Theorem~4]{schoeny2017codes}.

        \item[2.] When $b = 1$, the bound in Theorem~\ref{thm1: upper bound on M_q(n,(t,b))} reduces to
        \begin{align*}
        M_q(n,(t,1)) & \le \frac{t!\, q^{n}}{(q - 1)^t n^t}(1 + o(1)),
        \end{align*}
        when both $q$ and $t$ are fixed constants. This recovers the asymptotic version of the bound by Kulkarni and Kiyavash~\cite{kulkarni2013nonasymptotic}.

        \item [3.] In \cite{LSYG26}, based on a characterization of the size of $I_t^b(\mathbf{x})$ for any $\mathbf{x}\in \Sigma_q^{n}$, Lan \emph{et al.} derived the following upper bound on $M_q(n,(t,b))$:
        $$
        \frac{q^{n + t}}{\sum_{i = 0}^t \binom{n + t}{i}(q - 1)^i}= \frac{q^{n+t} \, t!}{n^t (q-1)^t}\bigl(1+o(1)\bigr).$$
        Clearly, the bound in Theorem~\ref{thm1: upper bound on M_q(n,(t,b))} coincides with this bound for the case $t=1$. For constant $t \ge 2$, the bound in Theorem~\ref{thm1: upper bound on M_q(n,(t,b))} can be written as
        $$
        \frac{t!q^{n - tb + t}}{(q - 1)^t\left(n-2tb-\frac{(t-1)b}{q}\right)^t}(1 + o(1))
        = \frac{t!q^{n - tb + t}}{(q - 1)^t n^t}\left(c + o(1)\right),
        $$
        where $c=\left(1-\frac{(2q+1)tb}{qn}\right)^{t}$ is a constant depending on the ratio $\frac{tb}{n}$. This shows that the bound in Theorem~\ref{thm1: upper bound on M_q(n,(t,b))} improves upon that in~\cite{LSYG26} by a factor of $q^{-tb}(c+o(1))$.
    \end{itemize}
\end{remark}

\subsection{A Sphere-Packing Upper Bound}

In this subsection, we present a sphere-packing upper bound on $M_q(n,(t,b))$, which coincides asymptotically with the bound in Theorem~\ref{thm1: upper bound on M_q(n,(t,b))} but improves it for certain parameter regimes. We begin by establishing a lower bound on $|D_t^b(\mathbf{x})|$ for any $\mathbf{x} \in \Sigma_q^n$, which generalizes the bound in \cite{hirschberg1999bounds} corresponding to $b = 1$.

\begin{theorem}\label{thm: improved lb on D_t^b}
    Let $n$, $t$, and $b$ be positive integers satisfying $n\geq tb+1$. Then, for any $\mathbf{x}\in \Sigma_q^n$, it holds that 
    $$
    |D_t^b(\mathbf{x})|\geq \sum_{s = 0}^t \binom{|U_b(\mathbf{x})| - (t - 1)b - 1}{s}.
    $$
\end{theorem}

\begin{IEEEproof}
    Given $\mathbf{x} \in \Sigma_q^n$, we first show that for every $2 \le s \le t$ and any indices $1 \le i_1 < \cdots < i_s \le n - (t - s)b - b$ satisfying $i_1, \ldots, i_s \in U_b(\mathbf{x})$ and $i_{j+1} - i_j > b$ for all $j \in [s-1]$, the compact representation
    \begin{equation}\label{eq: max rep of y}
    \mathrm{Del}_b(\mathbf{x}, (i_1,1), \ldots, (i_s,1), (n - (t - s)b + 1, t - s))
    \end{equation}
    is maximal for some sequence in $D_t^b(\mathbf{x})$.

    Let $\mathbf{y} \in D_t^b(\mathbf{x})$ be the sequence corresponding to \eqref{eq: max rep of y}, and suppose it admits another compact representation 
    $$
    \mathrm{Del}_b(\mathbf{x}, (i_1', c_1), \ldots, (i_{s'}', c_{s'}))
    $$ 
    that precedes \eqref{eq: max rep of y}. Let $\ell$ be the largest index such that $(i_j', c_j) = (i_j,1)$ for all $j \le \ell$.
    
    If $\ell < s$, then by definition of the order, we have either $i_{\ell+1}'>i_{\ell+1}$ or $i_{\ell+1}'=i_{\ell+1}$ and $1>c_{\ell+1}$. Clearly, since $c_{i}\geq 1$ holds for all $i\in [s']$, we obtain $i_{\ell+1}'> i_{\ell+1}$. Thus, by representation \eqref{eq: max rep of y}, we have 
    $$
    \mathbf{y} = (\ldots, x_{i_{\ell+1}-1},\, x_{i_{\ell+1} + b}, \ldots),
    $$
    which implies that $y_{i_{\ell+1}-\ell b}=x_{i_{\ell+1} + b}$. However, by representation $\mathrm{Del}_{b}(\mathbf{x}, (i_1', c_1), \ldots, (i_{s'}', c_{s'}))$, we also have
    $$
    \mathbf{y} = (\ldots, x_{i_{\ell+1}-1},\, x_{i_{\ell+1}}, \ldots, x_{i_{\ell+1}'-1,}, x_{i_{\ell+1}'+c_{\ell+1}b}, \ldots ),
    $$
    which implies that $y_{i_{\ell+1}-\ell b}=x_{i_{\ell+1}}$. This yields a contradiction since $x_{i_{\ell+1}}\neq x_{i_{\ell+1}+b}$ by $i_{\ell+1}\in U_{b}(\mathbf{x})$.

    If $\ell = s$, then $(i_j', c_j) = (i_j,1)$ for all $j \in [s]$. Moreover, the ordering condition implies either $i_{s+1}' > n - (t - s)b + 1$ or $i_{s+1}' = n - (t - s)b + 1$ with $c_{s+1} < t - s$, both of which lead to $\sum_{j=1}^{s'} c_j < t$, contradicting \eqref{eq: representation cond}. Hence, \eqref{eq: max rep of y} is maximal.

    Next, we complete the proof by bounding the number of representations of form \eqref{eq: max rep of y}.

    Notice that for any $s \leq t$, there are at least $|U_b(\mathbf{x})| - (t - s)b - 1$ valid positions in $U_b(\mathbf{x})$ not exceeding $n - (t - s)b - b$, yielding at least $\binom{|U_b(\mathbf{x})| - (t - 1)b - 1}{s}$ ways to choose $i_1, \ldots, i_s$. Summing over all $s \leq t$, we obtain at least $\sum_{s=0}^{t} \binom{|U_b(\mathbf{x})| - (t - 1)b - 1}{s}$ subsequences with distinct representations in $D_t^b(\mathbf{x})$. This establishes the desired lower bound on $D_t^b(\mathbf{x})$, which completes the proof.
\end{IEEEproof}

With Theorem \ref{thm: improved lb on D_t^b}, we can derive another non-asymptotic upper bound on $M_q(n,(t,b))$.

\begin{theorem}\label{thm: upper_bound_levenshtein}
    For positive integers $q\geq 2$, $n$ ,$t$, $b$, $r$ satisfying $n\geq tb$ and $n - b + 1\geq r + 1\geq t\geq 1$, it holds that
    \begin{align*}
        M_q(n,(t,b)) & \leq \frac{q^{n - tb}}{\sum_{s=0}^t \binom{r - (t - 1)b}{s}} + \sum_{i = 0}^{r - 1} \binom{n - b}{i}q^b(q-1)^i.
    \end{align*}
\end{theorem}
\begin{IEEEproof}
    Let $\cC\subseteq \Sigma_q^n$ be a $(t,b)$-burst-deletion-correcting code. For an integer $r$ satisfying $t-1\leq r\leq n-b$, we split $\cC$ into two subsets $\cC_1$ and $\cC_2$, where $\cC_1 = \{\mathbf{x}\in \cC : |U_b(\mathbf{x})| \leq r\}$ and $\cC_2 = \{\mathbf{x}\in \cC : |U_b(\mathbf{x})| > r\}$. Then, we have $|\cC| = |\cC_1| + |\cC_2|$. We next proceed to upper bound $|\cC_1|$ and $|\cC_2|$ separately. 

    For $|\cC_1|$, since $\cC_1\subseteq \{\mathbf{x}\in \Sigma_q^n : |U_b(\mathbf{x})| \leq r\}$, it suffices to determine the size of the set $\{\mathbf{x}\in \Sigma_q^n : |U_b(\mathbf{x})| \leq r\}$. Note that
    \begin{align*}
        \bigl|\{\mathbf{y}\in \Sigma_q^{n} : |U_b(\mathbf{y})| = r\}\bigr|
        = \binom{n-b}{r-1} q^{b} (q-1)^{r-1}.
    \end{align*}
    Then, it holds that
    \begin{align*}
        \bigl|\{\mathbf{x}\in \Sigma_q^n : |U_b(\mathbf{x})| \leq r\}\bigr| & = \sum_{i = 0}^{r - 1} \bigl|\{\mathbf{y}\in \Sigma_q^{n} : |U_b(\mathbf{y})| = i + 1\}\bigr| \\
        & = \sum_{i = 0}^{r - 1} \binom{n - b}{i}q^b(q-1)^i.
    \end{align*}

    By Theorem \ref{thm: improved lb on D_t^b}, for any $\mathbf{x}\in \cC_2$, we have $|D_t^b(\mathbf{x})|\geq \sum_{s=0}^t \binom{r + 1 - (t - 1)b - 1}{s} = \sum_{s=0}^t \binom{r - (t - 1)b}{s}$. Since $\cC$ is a $(t,b)$-burst-deletion-correcting code, it holds that $D_t^b(\mathbf{x})\cap D_t^b(\mathbf{x}') = \emptyset$ for any distinct $\mathbf{x}, \mathbf{x}'\in \cC_2$. Hence, we have
    \begin{align*}
        |\cC_2| & \leq \frac{|\bigcup_{\mathbf{x}\in \cC_2} D_t^b(\mathbf{x})|}{\min_{\mathbf{x}\in \cC_2} |D_t^b(\mathbf{x})|} = \frac{\sum_{\mathbf{x}\in \cC_2} |D_t^b(\mathbf{x})|}{\min_{\mathbf{x}\in \cC_2} |D_t^b(\mathbf{x})|} \leq \frac{q^{n - tb}}{\sum_{s=0}^t \binom{r - (t - 1)b}{s}}.
    \end{align*}
    In total, we have $$|\cC| = |\cC_1| + |\cC_2| \leq \frac{q^{n - tb}}{\sum_{s=0}^t \binom{r - (t - 1)b}{s}} + \sum_{i = 0}^{r - 1} \binom{n - b}{i}q^b(q-1)^i,$$ which completes the proof.
\end{IEEEproof}

\begin{remark}\label{rmk: comparison with LP bound}
    \begin{itemize}
        \item [1.] When $b = 1$, the upper bound on $M_q(n,(t,b))$ in Theorem~\ref{thm: upper_bound_levenshtein} recovers Levenshtein's upper bound for $t$-deletion-correcting codes.
        
        \item [2.] Using similar notation and proof techniques as in Theorem~\ref{thm1: upper bound on M_q(n,(t,b))}, by taking $r = \frac{q - 1}{q}(n - (t+1)b) - \sqrt{2t(n-b)\ln (n-b)} + 1$, the upper bound on $M_q(n,(t,b))$ in Theorem~\ref{thm: upper_bound_levenshtein} implies that
        \begin{align*}
        M_q(n,(t,b)) 
        &\leq \frac{q^{n - tb}}{\sum_{s=0}^t \binom{r - (t - 1)b}{s}} + q^b \mathrm{Vol}_q(r - 1, n - b) \\
        & \leq \frac{q^{n - tb}}{\binom{r - (t - 1)b}{t}} + q^{b + H_q\!\left(\frac{q-1}{q} - \frac{(q-1)tb}{q(n-b)} - \sqrt{\frac{2t\ln (n-b)}{n-b}}\right)(n - b)} \\
        & = \frac{t!\, q^{n - tb + t}}{(q-1)^{t}\left(n - 2tb - \frac{(t-1)b}{q}\right)^{t}} + O\!\left(\frac{q^n}{(n - b)^{qt}}\right) \\
        & = \frac{t!\, q^{n - tb + t}}{(q-1)^{t}\left(n - 2tb - \frac{(t-1)b}{q}\right)^{t}}(1 + o(1)),
        \end{align*}
        when both $q$ and $t$ are constants and $n \geq tb + 1$ is sufficiently large. This asymptotically coincides with the bound in Theorem~\ref{thm1: upper bound on M_q(n,(t,b))}.
    \end{itemize}
\end{remark}

Figure~\ref{fig: comparison of upper bounds 1} compares the non-asymptotic upper bounds on $M_q(n,(t,b))$ given in Theorem~\ref{thm: non asymptotic upper bound 1} and Theorem~\ref{thm: upper_bound_levenshtein} as well as the upper bound in \cite{LSYG26} for $q = 3$, $t = 3$, and $b = 2$. It can be observed that the bound in Theorem~\ref{thm: upper_bound_levenshtein} is tighter than that in Theorem~\ref{thm: non asymptotic upper bound 1} for small values of $n$.

\begin{figure}[h!]
    \centering
    \includegraphics[width=0.5\textwidth]{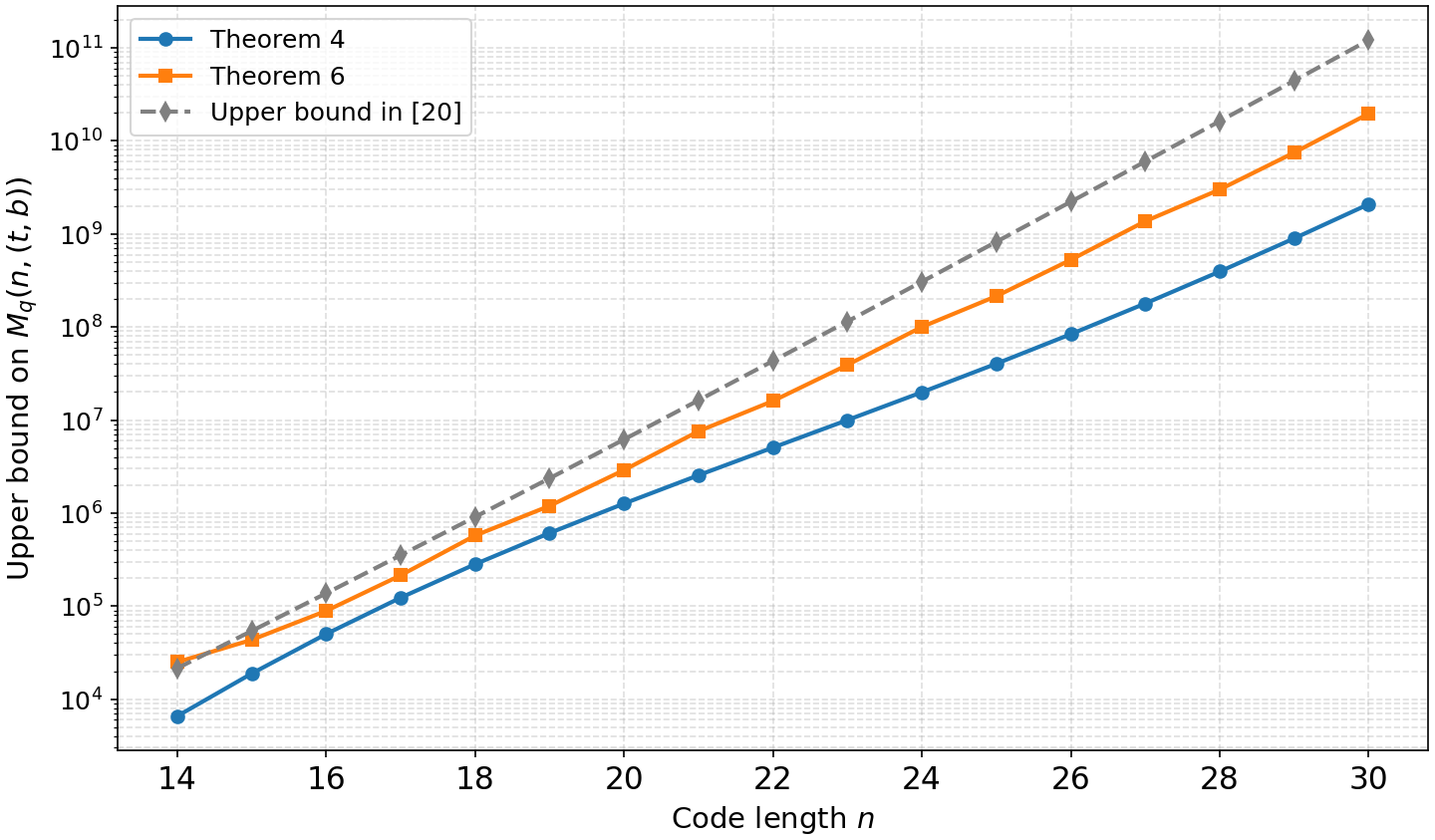}
    \caption{Comparison of the upper bounds on $M_q(n,(t,b))$ in Theorem~\ref{thm: non asymptotic upper bound 1}, Theorem~\ref{thm: upper_bound_levenshtein}, and the upper bound in \cite{LSYG26} for $q = 3$, $t = 3$, $b = 2$.}\label{fig: comparison of upper bounds 1}
\end{figure}

\subsection{Combinatorial Bounds on $M_q(n,(t,b))$}

In this subsection, we provide upper and lower bounds on $M_q(n,(t,b))$ for the regime where $q$ is large, which we refer to as combinatorial bounds. As a consequence, we obtain an asymptotic formula for $M_q(n,(t,b))$ when $q$ is sufficiently large, as stated in the following theorem.

\begin{theorem}\label{thm: combinatorial_bd_large_q}
    Let $n$, $t$, and $b$ be fixed positive integers satisfying $n \ge tb + 1$. Then,
    $$
    M_q(n,(t,b)) = \begin{cases}
        q^{b}, &\text{if }t = \frac{n}{b} - 1,\\
        \frac{q^{n - tb}}{\binom{n - tb + t}{t}}(1-o(1)), &\text{if }t < \frac{n}{b} - 1.
        \end{cases}
    $$
    as $q \to \infty$.
\end{theorem}

We begin by establishing the upper bound. More precisely, we have the following result.

\begin{theorem}\label{thm: upper_bound_2}
    For any $(t, b)$-burst deletion correcting code $\cC\subseteq \Sigma_q^n$ such that $b\mid n$ and $q\geq n-tb$, it holds that 
    $$
        |\cC|\leq \begin{cases}
        q^{b}, &\text{if }t = \frac{n}{b} - 1,\\
        q^{n - tb} - \binom{q}{n - tb}(n - tb)! + \frac{\binom{q}{n - tb}(n - tb)!}{\binom{n - t(b - 1)}{t}}, &\text{if }t < \frac{n}{b} - 1.
        \end{cases}
    $$
\end{theorem}

For the proof of Theorem~\ref{thm: upper_bound_2} we need the following lemma.

\begin{lemma}\label{lem:same_symbol}
    When $t < \frac{n}{b} - 1$, for any sequence $\mathbf{x} \in \Sigma_q^n$ that contains two identical symbols, there exists a sequence $\mathbf{y} \in D_t^b(\mathbf{x})$ such that $\mathbf{y}$ also contains two identical symbols.
\end{lemma}

\begin{IEEEproof}
    Without loss of generality, we assume that there exist $1\leq i_1 < i_2\leq n$ with $x_{i_1} = x_{i_2}$ such that $i_1\in [j_1b + 1, j_1b + b]$ and $i_2\in [j_2b + 1, j_2b + b]$ for some $j_1, j_2$. 
    
    When $j_1 = j_2$, consider the sequence
    \begin{align*}
        \mathbf{y} = \mathrm{Del}_{b}(\mathbf{x}, &\{1,b+1,\ldots,(j_1-1)b+1,\ldots,(j_1 + 1)b + 1, \ldots, tb + 1\}).
    \end{align*} 
    As $t < \frac{n}{b} - 1$, $\mathbf{y}$ is well-defined. Then, since $x_{i_1}$ and $x_{i_2}$ are preserved in $\mathbf{y}$, we obtain a sequence $\mathbf{y}\in D_t^b(\mathbf{x})$ that contains two identical symbols.

    When $j_1 \neq j_2$, consider the sequence
    \begin{align*}
        \mathbf{y} = \mathrm{Del}_{b}(\mathbf{x}, &\{1, \ldots, (j_1 - 1)b + 1, (j_1 + 1)b + 1, \ldots, \\
        &~(j_2 - 1)b + 1, (j_2 + 1)b + 1,\ldots, (t + 1)b + 1\}).
    \end{align*} 
    Similarly, since $t < \frac{n}{b} - 1$, the sequence $\mathbf{y}$ is well defined, and because $x_{i_1}$ and $x_{i_2}$ are preserved, it contains two identical symbols.
\end{IEEEproof}

Next, we present the proof of Theorem \ref{thm: upper_bound_2}.

\begin{IEEEproof}[Proof of Theorem \ref{thm: upper_bound_2}]
    For any $(t, b)$-burst deletion correcting code $\cC$, we partition $\cC$ into subsets $\cC_1, \ldots, \cC_n$ such that 
    $$
    \cC_i \triangleq \{\mathbf{x}\in \cC: \text{there are }i\text{ symbols in }\mathbf{x}\}\footnote{We define $\cC_i=\emptyset$ when $i>q$.}.
    $$ 

    When $t = \frac{n}{b} - 1$, a codeword $\mathbf{x}\in \cC$ will become a sequence of length $b$ after $t$ $b$-burst deletions. As the sequences after $t$ $b$-burst deletions are all distinct for different codewords, we have that $|\cC|\leq q^b$. 
    
    When $t < \frac{n}{b} - 1$, Lemma~\ref{lem:same_symbol} implies that for any $\mathbf{x} \in \mathcal{C}_i$ with $i < n$, there exists at least one sequence in $D_t^b(\mathbf{x})$ that contains two identical symbols. Observe that the number of sequences in $\Sigma_q^{\,n - tb}$ consisting of pairwise distinct symbols is $\binom{q}{n - tb}(n - tb)! = \prod_{i=0}^{n - tb - 1} (q - i)$. Therefore, since the deleted sequences corresponding to different codewords are distinct, it follows that 
    \begin{equation}\label{eq1_comb_bound}
    \sum_{i = 1}^{n - 1} |\mathcal{C}_i| \le q^{\,n - tb} - \prod_{i=0}^{n - tb - 1} (q - i).
    \end{equation}

    Furthermore, since any $\mathbf{x} \in \mathcal{C}_n$ consists of pairwise distinct symbols by definition, we have that $U_b(\mathbf{x}) = \{n - b + 1\}$. Consequently, for any $\mathbf{x} \in \mathcal{C}_n$, any two distinct choices of $t$ positions at which $b$-burst deletions occur lead to distinct subsequences in $D_t^b(\mathbf{x})$. Since there are $\binom{n - t(b - 1)}{t}$ different ways to choose the positions of the $t$ $b$-burst deletions, it follows that
    $$
    |D_t^b(\mathbf{x})| = \binom{n - t(b - 1)}{t}
    $$
    for every $\mathbf{x} \in \mathcal{C}_n$.
    Moreover, since the deletion balls of distinct codewords in $\mathcal{C}$ are pairwise disjoint, we have
    \begin{equation}\label{eq2_comb_bound}
        |\mathcal{C}_n|\le \frac{\prod_{i=0}^{n - tb - 1} (q - i)}{\binom{n - t(b - 1)}{t}}.
    \end{equation}
    Therefore, combining \eqref{eq1_comb_bound} and \eqref{eq2_comb_bound} together, we obtain
    \begin{align*}
    |\mathcal{C}|
    &= \sum_{i=1}^n |\mathcal{C}_i| \\
    &\le 
    q^{\,n - tb}
    - \prod_{i=0}^{n - tb - 1} (q - i)
    \left(1-\binom{n - t(b - 1)}{t}^{-1}\right).
    \end{align*}
    This completes the proof.
\end{IEEEproof}

\begin{remark}\label{rmk: comparision with Bours' result}
    In \cite{bours1995construction}, Bours established the following upper bound on $M_q(n,(t=n-2,b=1))$: for positive integers $n \ge 2$ and $q \ge 2$, any $(n-2,1)$-deletion-correcting code $\cC \subseteq \Sigma_q^n$ satisfies
    \begin{equation}\label{eq: bours' upp bd}
        |\cC| \le q + \left\lfloor \frac{q}{n} \left\lfloor \frac{2(q - 1)}{n - 1} \right\rfloor \right\rfloor.
    \end{equation}
    Note that for $b > 1$ with $b \mid n$, a $(t,b)$-burst-deletion-correcting code over $\Sigma_q^n$ can be viewed as a $(t,1)$-deletion-correcting code over $\Sigma_{q^b}^{\,n/b}$. Hence, \eqref{eq: bours' upp bd} implies that
    \begin{equation}\label{eq1: comparision with Bours' result}
        M_q(n,(t,b)) \le q^b + \left\lfloor \frac{q^b}{n/b} \left\lfloor \frac{2(q^b - 1)}{n/b - 1} \right\rfloor \right\rfloor 
    = \frac{2b^2 q^{2b}}{n(n - b)} + o(q^{2b}),
    \end{equation}
    as $q \to \infty$.
    
    One can easily check that the upper bound in Theorem \ref{thm: upper_bound_2} reduces to the upper bound \eqref{eq: bours' upp bd} when $b = 1$ and $t = n - 2$. Moreover, when $b\geq 2$ and $t = \frac{n}{b} - 2$, the upper bound in Theorem \ref{thm: upper_bound_2} becomes 
    \begin{equation}\label{eq2: comparision with Bours' result}
        q^{2b} - \binom{q}{2b}(2b)! + \frac{\binom{q}{2b}(2b)!}{\binom{n/b + 2b - 2}{n/b - 2}} = \frac{q^{2b}}{\binom{n/b + 2b - 2}{2b}} + o(q^{2b})
    \end{equation}
    as $q\rightarrow \infty$. Note that \begin{align*}
        \binom{n/b + 2b - 2}{2b} &= \frac{(n/b + 2b - 2)\cdots (n/b - 1)}{(2b)!}\\
        &=\frac{(n/b + 2b - 2)(n/b + 2b - 3)}{2}\cdot \prod_{i = 4}^{2b+1}\frac{n/b + 2b - i}{i - 1}.
    \end{align*}
    When $b\geq 2$ and $n\geq 4b$, we have 
    $\frac{(n + 2b^2 - 2b)(n + 2b^2 - 3b)}{2b^2} \geq \frac{n(n - b)}{2b^2}$ and $\prod_{i = 4}^{2b+1}\frac{(n/b + 2b - i)}{(i - 1)}\geq 1$. Hence, we have $\binom{n/b + 2b - 2}{2b}^{-1} \le \frac{2b^2}{n(n - b)}$, which implies that the bound in \eqref{eq2: comparision with Bours' result} is tighter than that in \eqref{eq1: comparision with Bours' result}.
\end{remark}

%\XLcom{Please note that the above approximation is obtained under the assumption that $q \to \infty$. We believe that a more detailed discussion is needed. In the above estimation, we treat $\frac{q^{2b}}{\binom{n/b + 2b - 2}{2b}}$ as the dominant term, which only holds when $|\mathcal{C}_n|$ is nonempty, and this further implies that $q \geq n$. Otherwise, when $q$ is fixed and $n \to \infty$, the above upper bound becomes $q^{2b}(1 - \prod_{i=0}^{2b-1}\left(1 - \frac{i}{q}\right))$. (I wonder whether such a bound is trivial)}

%\XLcom{In \cite{kong2025combinatorial}, we formulate a linear program based on observations about the quantity $\sum_{i\in [n-1]} |\mathcal{C}_i|$ and other summations used in the proof there. I believe that similar linear constraints can be derived in the present setting, which would allow us to construct a corresponding linear program and search for feasible solutions.}

Next, we present the proof of the lower bound on $M_q(n,(t,b))$ in Theorem \ref{thm: combinatorial_bd_large_q}. We first introduce some extra notation and preliminary results from hypergraph theory.

A hypergraph is called $\ell$-bounded if every hyperedge contains at most $\ell$ vertices. The degree of a vertex $v$ in $\mathcal{H}$, denoted by $\deg_{\mathcal{H}}(v)$, is the number of hyperedges containing $v$. The maximum degree of $\mathcal{H}$, denoted by $\Delta(\mathcal{H})$, is the maximum degree over all vertices in $\mathcal{H}$. The codegree of two distinct vertices $u$ and $v$ in $V(\mathcal{H})$, denoted by $\operatorname{cod}(u, v)$, is the number of hyperedges containing both $u$ and $v$. The codegree of $\mathcal{H}$, denoted by $\operatorname{cod}(\mathcal{H})$, is the maximum codegree over all pairs of distinct vertices in $\mathcal{H}$. 

In \cite{liu2024near}, the authors proved the following lower bound on the matching number of a hypergraph.

\begin{lemma}(\cite[Lemma II.4]{liu2024near})\label{lem:liu2024near}
Let $\mathcal{H}$ be an $\ell$-bounded hypergraph such that $\frac{\operatorname{cod}(\mathcal{H})}{\Delta(\mathcal{H})} = o(1)$ as $|V(\mathcal{H})| \to \infty$. Then,
$$
\nu(\mathcal{H}) > (1 - o(1)) \frac{|E(\mathcal{H})|}{\Delta(\mathcal{H})}.
$$
\end{lemma}

Recall that the hypergraph $\mathcal{H}_{q,n,t}^{b}$ has vertex set $V(\mathcal{H}_{q,n,t}^{b}) \triangleq \Sigma_q^{n-tb}$ and hyperedge set $E(\mathcal{H}_{q,n,t}^{b}) \triangleq \{ D_t^b(\mathbf{x}) : \mathbf{x} \in \Sigma_q^{n} \}$. We next establish several properties of the hypergraph $\mathcal{H}_{q,n,t}^{b}$, which will be used to derive a lower bound on $M_q(n,(t,b))$.

\begin{lemma}\label{lem:property_of_H}
    The hypergraph $\mathcal{H}_{q,n,t}^{b}$ satisfies the following properties:
    \begin{itemize}
        \item[1.] $\mathcal{H}_{q,n,t}^{b}$ is $\binom{n - (t-1)(b - 1)}{t}$-bounded.
        \item[2.] The maximum degree of $\mathcal{H}_{q,n,t}^{b}$ is
        $$
        \Delta(\mathcal{H}_{q,n,t}^{b}) = q^{t(b - 1)} \sum_{i = 0}^{t} \binom{n - tb + t}{i}(q - 1)^i.
        $$
        \item[3.] The codegree of $\mathcal{H}_{q,n,t}^{b}$ is
        $$
        \operatorname{cod}(\mathcal{H}_{q,n,t}^{b}) = q^{t(b - 1)} \sum_{i = 0}^{t - 1} \binom{n - tb + t}{i}(q - 1)^i \bigl[1 - (-1)^{t - i}\bigr].
        $$
        \end{itemize}
\end{lemma}

\begin{IEEEproof}   
    For Property~1, note that for any $\mathbf{x} \in \Sigma_q^n$, the possible positions for $t$ distinct $b$-burst deletions are determined by a sequence of $t$ integers $i_1, i_2, \ldots, i_t$ in $[n]$ such that $i_{j+1} \geq i_j + b$ holds for every $1 \leq j \leq t - 1$. Thus, for any $\mathbf{x} \in \Sigma_q^n$, we have
    $$
    |D_t^b(\mathbf{x})| \leq \binom{n - (t - 1)(b - 1)}{t}.
    $$
    Recall that each hyperedge in $\mathcal{H}_{q,n,t}^{b}$ is uniquely associated with a sequence $\mathbf{x} \in \Sigma_q^n$, and its size is exactly $|D_t^b(\mathbf{x})|$. Hence, $\mathcal{H}_{q,n,t}^{b}$ is $\binom{n - (t - 1)(b - 1)}{t}$-bounded.

    For Property~2, note that for any $\mathbf{y} \in \Sigma_q^{n - tb}$, viewed as a vertex in $\mathcal{H}_{q,n,t}^{b}$, its degree satisfies $\deg(\mathbf{y}) = |I_t^b(\mathbf{y})|$. By \cite[Theorem~1]{LSYG26}, we know that
    $$
    |I_t^b(\mathbf{y})| = q^{t(b - 1)} \sum_{i = 0}^{t} \binom{n - tb + t}{i}(q - 1)^i.
    $$
    This yields the desired expression for $\Delta(\mathcal{H}_{q,n,t}^{b})$.
    
    For Property~3, note that for any two distinct sequences $\mathbf{y}, \mathbf{z} \in \Sigma_q^{n - tb}$, viewed as vertices in $\mathcal{H}_{q,n,t}^{b}$, their codegree satisfies $\operatorname{cod}(\mathbf{y}, \mathbf{z}) = \left| I_t^b(\mathbf{y}) \cap I_t^b(\mathbf{z}) \right|$. Moreover, 
    by \cite[Theorem~2]{LSYG26}, we know that
    $$
    \max_{\mathbf{y} \neq \mathbf{z}} \left| I_t^b(\mathbf{y}) \cap I_t^b(\mathbf{z}) \right|
    = q^{t(b - 1)} \sum_{i = 0}^{t - 1} \binom{n - tb + t}{i}(q - 1)^i \bigl[1 - (-1)^{t - i}\bigr].
    $$
    This gives $\operatorname{cod}(\mathcal{H}_{q,n,t}^{b})$.
\end{IEEEproof}

Next, using Lemmas~\ref{lem:liu2024near} and~\ref{lem:property_of_H}, we establish the following lower bound on $M_q(n,(t,b))$ for the regime where $n$, $t$, and $b$ are fixed and $q \to \infty$, showing that the upper bound in Theorem \ref{thm: upper_bound_2} is tight in this parameter regime.

\begin{IEEEproof}[Proof of Theorem \ref{thm: combinatorial_bd_large_q}]
    When $n$, $t$, and $b$ are fixed and $q \to \infty$, we have
    $$
    \binom{q}{n-tb}(n-tb)! = q^{n-tb}\bigl(1 - O(q^{-1})\bigr).
    $$
    Hence, by Theorem~\ref{thm: upper_bound_2},
    $$
    M_q(n,(t,b)) \le 
    \begin{cases}
    q^{b}, & \text{if } t = \frac{n}{b} - 1, \\[0.5ex]
    \dfrac{q^{n - tb}}{\binom{n - tb + t}{t}}\bigl(1 - o(1)\bigr), & \text{if } t < \frac{n}{b} - 1.
    \end{cases}
    $$
    For the case $t = \frac{n}{b} - 1$, note that
    $$
    \mathcal{C} = \{(\underbrace{\mathbf{a}, \ldots, \mathbf{a}}_{n/b}) : \mathbf{a} \in \Sigma_q^b\}
    $$
    is an $(n/b - 1,b)$-burst-deletion-correcting code, which achieves the bound. This establishes the result for the case $t = \frac{n}{b} - 1$.
    
    It remains to prove the lower bound for the case $t < \frac{n}{b} - 1$.
    
    By Lemma~\ref{lem:property_of_H}, we have 
    $$
    \Delta(\mathcal{H}_{q,n,t}^{b})=q^{t(b-1)}\binom{n-tb+t}{t}(q-1)^{t}(1+o(1)),
    $$
    and
    $$
    \operatorname{cod}(\mathcal{H}_{q,n,t}^{b})= q^{t(b-1)}\binom{n-tb+t}{t-1}(q-1)^{t-1}(2+o(1)),
    $$
    as $q\rightarrow \infty$. This implies that $\operatorname{cod}(\mathcal{H}_{q,n,t}^{b}) / \Delta(\mathcal{H}_{q,n,t}^{b}) = o(1)$ as $q \to \infty$. Then, by applying Lemma~\ref{lem:liu2024near}, it follows that
    $$
    \nu(\mathcal{H}_{q,n,t}^{b}) > (1 - o(1)) \frac{|E(\mathcal{H}_{q,n,t}^{b})|}{\Delta(\mathcal{H}_{q,n,t}^{b})}.
    $$
    Substituting $|E(\mathcal{H}_{q,n,t}^{b})| = q^n$ and the expression for $\Delta(\mathcal{H}_{q,n,t}^{b})$, we obtain
    \begin{align*}
    \nu(\mathcal{H}_{q,n,t}^{b})
    &> (1 - o(1)) \frac{q^n}{q^{t(b - 1)} \sum_{i = 0}^{t} \binom{n - tb + t}{i}(q - 1)^i} \\
    &= (1 - o(1)) \frac{q^{n - tb + t}}{\sum_{i = 0}^{t} \binom{n - tb + t}{i}(q - 1)^i}.
    \end{align*}
    Since
    $$
    \sum_{i = 0}^{t} \binom{n - tb + t}{i}(q - 1)^i
    = \binom{n - tb + t}{t} q^t + o(q^t),
    $$
    the result follows.
\end{IEEEproof}

\section{Improved on the Previous Upper Bound for $b=t=2$}\label{sec: improved ub t=b=2}

In this section, we investigate the size of a $(2, 2)$-burst deletion-correcting code $\mathcal{C}$ over $\Sigma_q^n$. Specifically, by analyzing the relationship between the subsequences in $D_{2}^{2}(\mathbf{x})$ and those in $D_{2}(\mathbf{x})$, we obtain an improved lower bound on $|D_{2}^{2}(\mathbf{x})|$. Consequently, this leads to an improved upper bound on $M_q(n,(2,2))$ upon Theorem~\ref{thm: non asymptotic upper bound 1} and Theorem~\ref{thm: upper_bound_levenshtein}.

We begin by introducing some necessary notation and then establish two preliminary lemmas. For an event $E$, let $\delta_E = 1$ if $E$ occurs and $\delta_E = 0$ otherwise. 
% For any two sequences $\mathbf{x} = (x_1, \ldots, x_n)$ and $\mathbf{y} = (y_1, \ldots, y_n)$ in $\Sigma_q^n$, let 
% $$
% \mathbf{x} \| \mathbf{y} = \left[\begin{array}{cccc}
%     x_1 & x_2 & \cdots & \\
%     y_1 & y_2
% \end{array}
% \right]
% $$ 
% denote their interleaving. 
For any sequence $\mathbf{x} \in \Sigma_q^n$, we use $\mathrm{r}_{\mathrm{last}}(\mathbf{x})$ to denote the length of the last run of $\mathbf{x}$.

\begin{lemma}\label{lem:without last bit}
    For a sequence $\mathbf{x} \in \Sigma_q^n$, the number of sequences in $D_2(\mathbf{x}) \setminus D_{1}(\mathbf{x}_{[n-1]})$ is given by
    \begin{align*}
    |D_2(\mathbf{x})| - |U_1(\mathbf{x}_{[n - r_{\mathrm{last}}(\mathbf{x})]})| 
    + \delta_{r_{\mathrm{last}}(\mathbf{x}) = 1} - 1.
    \end{align*}
\end{lemma}
\begin{IEEEproof}
    Note that $D_{1}(\mathbf{x}_{[n-1]})\subseteq D_{2}(\mathbf{x})$, so it suffices to show that 
    $$
    |D_{1}(\mathbf{x}_{[n-1]})|=|U_1(\mathbf{x}_{[n - r_{\mathrm{last}}(\mathbf{x})]})| 
    - \delta_{r_{\mathrm{last}}(\mathbf{x}) = 1} + 1.
    $$
    Recall that by Proposition~\ref{prop: max rep of seq in D_t^b(x) is unique}, every sequence in $D_{2}(\mathbf{x})$ admits a unique maximal representation. Moreover, for any sequence $\mathbf{y} \in D_{1}(\mathbf{x}_{[n-1]})$, its maximal representation in $D_{2}(\mathbf{x})$ is either of the form $\mathrm{Del}_1(\mathbf{x}, (i, 1), (n, 1))$ for some $i \in [n - r_{\mathrm{last}}(\mathbf{x})]$ with $x_i \neq x_{i+1}$, or of the form $\mathrm{Del}_1(\mathbf{x}, (n - 1, 2))$.
    
    Since the number of indices $i \in [n - r_{\mathrm{last}}(\mathbf{x})]$ such that $x_i \neq x_{i+1}$ is exactly $|U_1(\mathbf{x}_{[n - r_{\mathrm{last}}(\mathbf{x})]})|$, there are $|U_1(\mathbf{x}_{[n - r_{\mathrm{last}}(\mathbf{x})]})|$ distinct sequences in $D_{1}(\mathbf{x}_{[n-1]})$ whose maximal representations in $D_{2}(\mathbf{x})$ are of the form $\mathrm{Del}_1(\mathbf{x}, (i, 1), (n, 1))$. Moreover, exactly one sequence in $D_{1}(\mathbf{x}_{[n-1]})$ has maximal representation $\mathrm{Del}_1(\mathbf{x}, (n-1, 2))$ in $D_{2}(\mathbf{x})$.

    Since $n - r_{\mathrm{last}}(\mathbf{x}) \in U_1(\mathbf{x}_{[n - r_{\mathrm{last}}(\mathbf{x})]})$ by definition, when $r_{\mathrm{last}}(\mathbf{x}) = 1$, the representations $\mathrm{Del}_1(\mathbf{x}, (n - r_{\mathrm{last}}(\mathbf{x}), 1), (n, 1))$ and $\mathrm{Del}_1(\mathbf{x}, (n - 1, 2))$ correspond to the same sequence. Therefore, there are
    \begin{align*}
    |U_1(\mathbf{x}_{[n - r_{\mathrm{last}}(\mathbf{x})]})| 
    - \delta_{r_{\mathrm{last}}(\mathbf{x}) = 1} + 1
     \end{align*}
     distinct sequences in $D_{1}(\mathbf{x}_{[n-1]})$, as claimed. This yields the desired result and completes the proof.
\end{IEEEproof}

\begin{lemma}\label{lem:two_pairs}
    For any positive even integer $n$ and any sequence $\mathbf{x} \in \Sigma_q^n$ with array representation $\begin{bmatrix}
    \mathbf{x}^{(1)}; 
    \mathbf{x}^{(2)}
    \end{bmatrix},$
    it holds that
    \begin{align}
    |D_2^2(\mathbf{x})|
    &= |D_2(\mathbf{x}^{(1)})| - 1 \nonumber \\
    &\quad + |D_2(\mathbf{x}^{(2)})| 
    - |U_1(\mathbf{x}^{(2)}_{[n/2 - r_{\mathrm{last}}(\mathbf{x}^{(2)})]})| 
    + \delta_{r_{\mathrm{last}}(\mathbf{x}^{(2)}) = 1} \nonumber \\
    &\quad + \sum_{i \in U_1(\mathbf{x}^{(1)})} 
    \left( |U_1(D_1(\mathbf{x}^{(2)}, i))| 
    - \delta_{x^{(2)}_{i-1} \neq x^{(2)}_{i+1}} - 1 \right). \label{eq: D_2^2(x)}
    \end{align}
    Here, we define $\delta_{x^{(2)}_{i-1} \neq x^{(2)}_{i+1}} = 0$ when $i = 1$ or $i = n/2$.
\end{lemma}
\begin{IEEEproof}
    By Proposition \ref{prop: max rep of seq in D_t^b(x) is unique}, it suffices to count the number of distinct maximal representations of sequences in $D_2^2(\mathbf{x})$.

    For any $\mathbf{y} \in D_2^2(\mathbf{x})$, its maximal representation has the following two possible forms:
    \begin{itemize}
        \item [(i)] $\mathrm{Del}_2(\mathbf{x}, (i_1,1),(i_2,1))$ with $i_2>i_1+2$ and $x_{i_j}\neq x_{i_j+2}$ for $j\in\{1,2\}$;
        \item [(ii)] $\mathrm{Del}_2(\mathbf{x}, (i,2))$ with $x_i\neq x_{i+2}$ and $x_{i+2}\neq x_{i+4}$.
    \end{itemize}
    Next, we split the discussion according to the parity of the deletion positions.

    \textbf{Case 1.}
    If all deletion positions are odd (i.e., $i_1\equiv i_2\equiv 1 \pmod{2}$ or $i\equiv 1 \pmod{2}$), then both $2$-burst deletions occur in $\mathbf{x}^{(1)}$. When $\mathbf{y}$ has a maximal representation of form (i), the condition $x_{i_j}\neq x_{i_j+2}$ for $j\in \{1,2\}$ implies that 
    $$
    x^{(1)}_{(i_j+1)/2} \neq x^{(1)}_{(i_j+1)/2 + 1}, ~j \in \{1,2\}.
    $$
    When $\mathbf{y}$ has a maximal representation of form (ii), the conditions $x_i\neq x_{i+2}$ and $x_{i+2}\neq x_{i+4}$ imply that
    $$
    x^{(1)}_{(i+1)/2} \neq x^{(1)}_{(i+1)/2 + 1}, ~x^{(1)}_{(i+1)/2 + 1} \neq x^{(1)}_{(i+1)/2 + 2}.
    $$
    This establishes a one-to-one correspondence between 
    $$ 
    \left\{ \mathbf{y} \in D_2^2(\mathbf{x}) :~ \begin{array}{l} \mathbf{y}\ \text{has a maximal representation of the form} \\ \mathrm{Del}_2(\mathbf{x}, (i_1,1),(i_2,1)) \text{ with } i_1 \equiv i_2 \equiv 1 \pmod{2}, \\ \text{or of the form } \mathrm{Del}_2(\mathbf{x}, (i,2)) \text{ with } i \equiv 1 \pmod{2} \end{array} \right\} 
    $$ 
    and $D_2(\mathbf{x}^{(1)})$. Consequently, the number of such sequences is exactly $|D_2(\mathbf{x}^{(1)})|$.

    \textbf{Case 2.}
    If all deletion positions are even (i.e., $i_1\equiv i_2\equiv 0 \pmod{2}$ or $i\equiv 0 \pmod{2}$), then both $2$-burst deletions occur in $\mathbf{x}^{(2)}$. Moreover, the second $2$-burst deletion cannot happen at position $n/2$ of $\mathbf{x}^{(2)}$. When $\mathbf{y}$ has a maximal representation of form (i), the condition $x_{i_j}\neq x_{i_j+2}$ for $j\in \{1,2\}$ implies that 
    $$
    x^{(2)}_{i_j/2} \neq x^{(2)}_{i_j/2 + 1}, ~j \in \{1,2\}.
    $$
    When $\mathbf{y}$ has a maximal representation of form (ii), the conditions $x_i\neq x_{i+2}$ and $x_{i+2}\neq x_{i+4}$ imply that
    $$
    x^{(2)}_{i/2} \neq x^{(2)}_{i/2 + 1}, ~x^{(2)}_{i/2 + 1} \neq x^{(2)}_{i/2 + 2}.
    $$
    This establishes a one-to-one correspondence between 
    $$ 
    \left\{ \mathbf{y} \in D_2^2(\mathbf{x}) :~ \begin{array}{l} \mathbf{y}\ \text{has a maximal representation of the form} \\ \mathrm{Del}_2(\mathbf{x}, (i_1,1),(i_2,1)) \text{ with } i_1 \equiv i_2 \equiv 0 \pmod{2}, \\ \text{or of the form } \mathrm{Del}_2(\mathbf{x}, (i,2)) \text{ with } i \equiv 0 \pmod{2} \end{array} \right\} 
    $$ 
    and $D_2(\mathbf{x}^{(2)})\setminus D_{1}(\mathbf{x}^{(2)}_{[n/2-1]})$. Consequently, by Lemma \ref{lem:without last bit}, the number of such sequences is
    $$
    |D_2(\mathbf{x}^{(2)})| - |U_1(\mathbf{x}^{(2)}_{[n/2 - r_{\mathrm{last}}(\mathbf{x}^{(2)})]})| + \delta_{r_{\mathrm{last}}(\mathbf{x}^{(2)}) = 1} - 1.
    $$

    \textbf{Case 3.}
    If $i_1\not\equiv i_2 \pmod{2}$, assume w.l.o.g.\ that $i_1$ is odd. Then $x_{i_1}\neq x_{i_1+2}$ implies $(i_1+1)/2\in U_1(\mathbf{x}^{(1)})$. The second deletion (at an even position) corresponds to deleting an element in $\mathbf{x}^{(2)}$ after removing position $(i_1+1)/2$, yielding an index in $U_1(D_1(\mathbf{x}^{(2)}, (i_1+1)/2))$ excluding $(i_1-1)/2$ and the boundary position $n/2$. This establishes a one-to-one correspondence between 
    $$ 
    \left\{ \mathbf{y} \in D_2^2(\mathbf{x}) :~ \begin{array}{l} \mathbf{y}\ \text{has a maximal representation of the form} \\ \mathrm{Del}_2(\mathbf{x}, (i_1,1),(i_2,1)) \text{ with } i_1 \not\equiv i_2 \pmod{2} \end{array} \right\} 
    $$ 
    and pairs $(i,i')$ such that
    $$
    i\in U_1(\mathbf{x}^{(1)}),~ i'\in U_1(D_1(\mathbf{x}^{(2)}, i))\setminus\{n/2,i-1\}.
    $$
    Thus, the number of such sequences equals
    $$
    \sum_{i\in U_1(\mathbf{x}^{(1)})}\big(|U_1(D_1(\mathbf{x}^{(2)}, i))|- \delta_{x^{(2)}_{i-1}\neq x^{(2)}_{i+1}} - 1\big).
    $$
    
    Combining all cases yields the desired result.
\end{IEEEproof}

Next, we provide an example to illustrate how to determine the size of $D_{2}^{2}(\mathbf{x})$ using Lemma \ref{lem:two_pairs}.

\begin{example}
Let $\mathbf{x} = 00011011\in \Sigma_2^8$. Then we have $\mathbf{x}^{(1)} = 0011$ and $\mathbf{x}^{(2)} = 0101$. One can easily verify that $U_1(\mathbf{x}^{(1)}) = \{2,4\}$, $U_1(\mathbf{x}^{(2)}) = \{1,2,3,4\}$, and $r_{\mathrm{last}}(\mathbf{x}^{(2)})=1$. This implies that $|D_2(\mathbf{x}^{(1)})| = 3$, $|D_2(\mathbf{x}^{(2)})| = 4$, and $|U_1(\mathbf{x}^{(2)}_{[3]})| = 3$. Moreover, since $D_1(\mathbf{x}^{(2)}, 2) = 001$ and $D_1(\mathbf{x}^{(2)}, 4) = 010$, Lemma \ref{lem:two_pairs} implies that 
\begin{align*}
    |D_2^2(\mathbf{x})| =\; & |D_2(\mathbf{x}^{(1)})| - 1\\
    & + |D_2(\mathbf{x}^{(2)})| - |U_1(\mathbf{x}^{(2)}_{[n/2 - r_{\mathrm{last}}(\mathbf{x}^{(2)})]})| + \delta_{r_{\mathrm{last}}(\mathbf{x}^{(2)}) = 1}\\
    & + \sum_{i \in U_1(\mathbf{x}^{(1)})} \left( |U_1(D_1(\mathbf{x}^{(2)}, i))| - \delta_{x^{(2)}_{i-1} \neq x^{(2)}_{i+1}} - 1 \right)\\
    =\; & 3 - 1 + 4 - 3 + 1 + 1 + 2 = 7.
\end{align*}

On the other hand, one can verify that 
\begin{align*}
    D_2^2(\mathbf{x}) = \{1011, 0011, 0001, 0111, 0110, 0010, 0000\},
\end{align*}
which confirms the correctness of the above calculation.
\end{example}

For a positive integer $r$, define
\begin{equation}\label{eq: def V_r}
    V_r \triangleq
    \begin{cases}
        010101\ldots 01^{n / 2 - r + 1}, & \text{if $2\mid r$}; \\
        010101\ldots 010^{n / 2 - r + 1}, & \text{if $2\nmid r$}.
    \end{cases}
\end{equation}
For any even integer $n$, it was shown in \cite{li2026number} that for any $\mathbf{z}\in \Sigma_q^{n/2}$ with $r$ runs, we have $|D_2(\mathbf{z})|\geq |D_2(V_r)|$. In \cite{liron2015characterization}, $|D_2(V_r)|$ is characterized for all values of $n$ and $r$, as stated in the following lemma.

\begin{lemma}[Lemma 15, \cite{liron2015characterization}]\label{lem: LL}
For all integers $r$ with $0 < r \leq n$ and $n > 2$, let $u(n, r) = |D_2(V_r)|$. Then
\begin{align*}
    u(n, r) = \begin{cases}
    r, & \text{if } r = 1, 2, \\
    2, & \text{if } r = n = 3, \\
    d(n, 2), & \text{if } r = n, \\
    d(r, 2) + 1, & \text{otherwise},
    \end{cases}
\end{align*}
where $d(n, t) = \sum_{i = 0}^{t} \binom{n - t}{i}$. We set $d(n, 0) = 1$, and for $t < 0$, $d(n, t) = 0$.
\end{lemma}

Next, using the above preliminary results, we prove the following improved lower bound on the size of the deletion ball $D_2^2(\mathbf{x})$ for any $\mathbf{x} \in \Sigma_q^n$ with $2 \mid n$.

\begin{theorem}\label{thm: lower_bound_for_2}
    For any even integer $n \geq 4$, let $\mathbf{x} \in \Sigma_q^{n}$ be a sequence with array representation
    $\begin{bmatrix}
    \mathbf{x}^{(1)};
    \mathbf{x}^{(2)}
    \end{bmatrix}$
    such that the numbers of runs of $\mathbf{x}^{(1)}$ and $\mathbf{x}^{(2)}$ are $r_1$ and $r_2$, respectively. Then it holds that
    \begin{align*}
    |D_2^2(\mathbf{x})| &\geq |D_2(V_{r_1})| + |D_2(V_{r_2})| - \min\{r_1, r_2\}\\
    &\quad + (\max\{r_1, r_2\} - 2)(\min\{r_1, r_2\} - 1)\\
    &\quad + \delta_{r_1 = r_2} - \delta_{r_1 = r_2 = 1}.
    \end{align*}
    Furthermore, the lower bound is tight when $\mathbf{x}^{(1)} = V_{r_1}$ and $\mathbf{x}^{(2)} = V_{r_2}$.
\end{theorem}
\begin{IEEEproof}
    By Lemma \ref{lem:two_pairs}, it suffices to show that the right-hand side of \eqref{eq: D_2^2(x)} is lower bounded by the claimed result for all $\mathbf{x}^{(1)}, \mathbf{x}^{(2)} \in \Sigma_q^{n/2}$ with $r_1$ and $r_2$ runs, respectively.

    By \cite[Theorem 7]{li2026number}, we have
    \begin{align}\label{eq: lb of D_2(x)}
        |D_2(\mathbf{x}^{(1)})|\geq |D_2(V_{r_1})|~\text{and}~|D_2(\mathbf{x}^{(2)})|\geq |D_2(V_{r_2})|.
    \end{align}
    This yields $|D_2(V_{r_1})|$ and $|D_2(V_{r_2})|$ as lower bounds. Moreover, one can easily verify that 
    \begin{align}\label{eq: U_1(x_{n-r_L})}
        |U_1(\mathbf{x}^{(2)}_{[n/2 - r_{\mathrm{last}}(\mathbf{x}^{(2)})]})| = r_2 - 1.
    \end{align} 
     
    It remains to bound the following term:
    $$
    \sum_{i\in U_1(\mathbf{x}^{(1)})} \left(|U_1(D_1(\mathbf{x}^{(2)}, i))| - \delta_{x^{(2)}_{i - 1} \neq x^{(2)}_{i + 1}} - 1\right) + \delta_{r_{\mathrm{last}}(\mathbf{x}^{(2)}) = 1}.
    $$  
    For simplicity, we assume without loss of generality that $r_1 \geq r_2$; otherwise, we may consider the sequence with array representation
    $\begin{bmatrix}
     \mathbf{x}^{(2)};
     \mathbf{x}^{(1)}
    \end{bmatrix}.$

    We start with the simplest case: $r_2 = 1$. Since $\mathbf{x}^{(2)}$ consists of only one run, we have $|U_1(D_1(\mathbf{x}^{(2)}, i))| = 1$ and $\delta_{x^{(2)}_{i - 1} \neq x^{(2)}_{i + 1}} = 0$ for all $i \in [n/2]$. Moreover, since $n/2 \geq 2$, we also have $\delta_{r_{\mathrm{last}}(\mathbf{x}^{(2)}) = 1} = 0$. Thus, we conclude that
    \begin{align}\label{eq: r_2=1}
     \sum_{i \in U_1(\mathbf{x}^{(1)})} \left(|U_1(D_1(\mathbf{x}^{(2)}, i))| - \delta_{x^{(2)}_{i - 1} \neq x^{(2)}_{i + 1}} - 1\right) + \delta_{r_{\mathrm{last}}(\mathbf{x}^{(2)}) = 1} = 0.
    \end{align}

    Next, we focus on the case $r_2 \geq 2$.  For each $j \in [r_2]$, let $\ell_{j}$ denote the length of the $j$-th run in $\mathbf{x}^{(2)}$. Then, the positions of the $j$-th run of $\mathbf{x}^{(2)}$ are indexed by
    $$
    R_{j} \triangleq \left[\sum_{h=1}^{j-1}\ell_{h} + 1 : \sum_{h=1}^{j}\ell_{h}\right].
    $$
    In the following, for each $i \in U_1(\mathbf{x}^{(1)})$, we bound
    $$
    \theta_i \triangleq |U_1(D_1(\mathbf{x}^{(2)}, i))| - \delta_{x^{(2)}_{i - 1} \neq x^{(2)}_{i + 1}} - 1
    $$
    from below according to the position of $x^{(2)}_{i}$ in $\mathbf{x}^{(2)}$.
    
    \textbf{Case 1.} $x^{(2)}_{i}$ lies in the $j$-th run of $\mathbf{x}^{(2)}$, where $1 < j < r_2$. 

    If $|R_j|=1$, then we have
    \begin{equation}
    \left\{
    \begin{array}{ll}
    |U_1(D_1(\mathbf{x}^{(2)}, i))| = r_2 - 2 \ \text{and} \ \delta_{x^{(2)}_{i - 1} \neq x^{(2)}_{i + 1}} = 0, & \text{if } x^{(2)}_{i - 1} = x^{(2)}_{i + 1}; \\
    |U_1(D_1(\mathbf{x}^{(2)}, i))| = r_2 - 1 \ \text{and} \ \delta_{x^{(2)}_{i - 1} \neq x^{(2)}_{i + 1}} = 1, & \text{if } x^{(2)}_{i - 1} \neq x^{(2)}_{i + 1}.
    \end{array}
    \right.
    \end{equation}
    In both cases, we have $\theta_i = r_2 - 3$.

    If $|R_j|\geq 2$, then $|U_1(D_1(\mathbf{x}^{(2)}, i))| = r_2$. Moreover, if $x^{(2)}_{i}$ is the first or the last bit of the run, then $\delta_{x^{(2)}_{i - 1} \neq x^{(2)}_{i + 1}} = 1$; otherwise, $\delta_{x^{(2)}_{i - 1} \neq x^{(2)}_{i + 1}} = 0$. Therefore, there are at most two positions $i \in U_1(\mathbf{x}^{(1)})\cap R_j$ such that $\theta_i = r_2 - 2$, and for all other $i \in U_1(\mathbf{x}^{(1)})\cap R_j$, we have $\theta_i = r_2 - 1$.

    \textbf{Case 2.} $x^{(2)}_{i}$ lies in the first run of $\mathbf{x}^{(2)}$. 

    If $|R_1|=1$, then we have $|U_1(D_1(\mathbf{x}^{(2)}, i))| = r_2 - 1$ and $\delta_{x^{(2)}_{i - 1} \neq x^{(2)}_{i + 1}} = 0$. This yields $\theta_i = r_2 - 2$.

    If $|R_1|\geq 2$, then $|U_1(D_1(\mathbf{x}^{(2)}, i))| = r_2$. Moreover, if $x^{(2)}_{i}$ is the last bit of the run, then $\delta_{x^{(2)}_{i - 1} \neq x^{(2)}_{i + 1}} = 1$; otherwise, $\delta_{x^{(2)}_{i - 1} \neq x^{(2)}_{i + 1}} = 0$. Therefore, there is at most one position $i \in U_1(\mathbf{x}^{(1)})\cap R_1$ such that $\theta_i = r_2 - 2$, and for all other $i \in U_1(\mathbf{x}^{(1)})\cap R_1$, we have $\theta_i = r_2 - 1$.

    \textbf{Case 3.} $x^{(2)}_{i}$ lies in the last run of $\mathbf{x}^{(2)}$. 

    %For each $i \in U_1(\mathbf{x}^1)$ such that $i$ lies in the last run of $\mathbf{x}^2$, we analyze $|U_1(D_1^1(\mathbf{x}^2, i))| - \delta_{\mathbf{x}^2_{i - 1} \neq \mathbf{x}^2_{i + 1}} - 1 - \delta_{r_{l}(\mathbf{x}^2) = 1}$. We distinguish the following cases.

    If $|R_{r_2}|=1$, then we have $|U_1(D_1(\mathbf{x}^{(2)}, i))| = r_2 - 1$, $\delta_{x^{(2)}_{i - 1} \neq x^{(2)}_{i + 1}} = 0$, and $\delta_{r_{\mathrm{last}}(\mathbf{x}^{(2)}) = 1} = 1$. Thus, $\theta_i = r_2 - 2$.

    If $|R_{r_2}|=1$, then $|U_1(D_1(\mathbf{x}^{(2)}, i))| = r_2$ and $\delta_{r_{\mathrm{last}}(\mathbf{x}^{(2)}) = 1} = 0$. Moreover, if $x^{(2)}_{i}$ is the first bit of the run, then $\delta_{x^{(2)}_{i - 1} \neq x^{(2)}_{i + 1}} = 1$; otherwise, $\delta_{x^{(2)}_{i - 1} \neq x^{(2)}_{i + 1}} = 0$. Therefore, there is at most one position $i \in U_1(\mathbf{x}^{(1)})\cap R_{r_2}$ such that $\theta_i = r_2 - 2$, and for all other $i \in U_1(\mathbf{x}^{(1)})\cap R_{r_2}$, we have $\theta_i = r_2 - 1$.

    Based on the above three cases, we can summarize  that
    \begin{align}
    \sum_{i \in U_1(\mathbf{x}^{(1)}) \cap R_{j}} \theta_i 
    & \geq \begin{cases}
        |U_1(\mathbf{x}^{(1)}) \cap R_{j}|(r_2 - 3), & \text{if } |U_1(\mathbf{x}^{(1)}) \cap R_{j}| \leq 1;\\
        (|U_1(\mathbf{x}^{(1)}) \cap R_{j}| - 2)(r_2 - 1) + 2(r_2 - 2), & \text{otherwise};
    \end{cases} \label{eq: theta_i sum1.0} \\
    & = |U_1(\mathbf{x}^{(1)}) \cap R_{j}|(r_2 - 1) - 2\delta_{|U_1(\mathbf{x}^{(1)}) \cap R_{j}| \neq 0}, \label{eq: theta_i sum1}
    \end{align}
    for $1 < j < r_2$, and
    \begin{align}
    \sum_{i \in U_1(\mathbf{x}^{(1)}) \cap R_{1}} \theta_i 
    & \geq \begin{cases}
        |U_1(\mathbf{x}^{(1)}) \cap R_{1}|(r_2 - 2), & \text{if } |U_1(\mathbf{x}^{(1)}) \cap R_{1}| \leq 1;\\
        (|U_1(\mathbf{x}^{(1)}) \cap R_{1}| - 1)(r_2 - 1) + (r_2 - 2), & \text{otherwise};
    \end{cases} \label{eq: theta_i sum2.0} \\
    & = |U_1(\mathbf{x}^{(1)}) \cap R_{1}|(r_2 - 1) - \delta_{|U_1(\mathbf{x}^{(1)}) \cap R_{1}| \neq 0}. \label{eq: theta_i sum2}
    \end{align}
    Moreover, based on the discussion in \textbf{Case 3}, regardless of whether $r_{\mathrm{last}}(\mathbf{x}^{(2)}) = 1$ or not, we always have
    $$
    \theta_{n/2} + \delta_{r_{\mathrm{last}}(\mathbf{x}^{(2)}) = 1} = r_2 - 1.
    $$
    Since $n/2 \in U_1(\mathbf{x}^{(1)}) \cap R_{r_2}$ always holds, this further implies that
    \begin{align}
    \sum_{i \in U_1(\mathbf{x}^{(1)}) \cap R_{r_2}} \theta_i + \delta_{r_{\mathrm{last}}(\mathbf{x}^{(2)}) = 1}
    & \geq \begin{cases}
        |U_1(\mathbf{x}^{(1)}) \cap R_{r_2}|(r_2 - 1), & \text{if } |U_1(\mathbf{x}^{(1)}) \cap R_{r_2}| \leq 1;\\
        (|U_1(\mathbf{x}^{(1)}) \cap R_{r_2}| - 1)(r_2 - 1) + (r_2 - 2), & \text{otherwise};
    \end{cases} \label{eq: theta_i sum3.0} \\
    & = |U_1(\mathbf{x}^{(1)}) \cap R_{r_2}|(r_2 - 1) - \delta_{|U_1(\mathbf{x}^{(1)}) \cap R_{r_2}| \geq 2}. \label{eq: theta_i sum3}
    \end{align}
    Therefore, since $|U_1(\mathbf{x}^{(1)})| = r_1$, by combining \eqref{eq: theta_i sum1}, \eqref{eq: theta_i sum2}, and \eqref{eq: theta_i sum3}, we obtain the following lower bound:
    \begin{align}
    & \sum_{i \in U_1(\mathbf{x}^{(1)})} \left(|U_1(D_1(\mathbf{x}^{(2)}, i))| - \delta_{x^{(2)}_{i - 1} \neq x^{(2)}_{i + 1}} - 1\right) + \delta_{r_{\mathrm{last}}(\mathbf{x}^{(2)}) = 1} 
    \nonumber \\
    =\; & \sum_{i \in U_1(\mathbf{x}^{(1)}) \cap R_1} \theta_i 
    + \sum_{j=2}^{r_2-1} \sum_{i \in U_1(\mathbf{x}^{(1)}) \cap R_j} \theta_i 
    + \sum_{i \in U_1(\mathbf{x}^{(1)}) \cap R_{r_2}} \theta_i 
    + \delta_{r_{\mathrm{last}}(\mathbf{x}^{(2)}) = 1}
    \nonumber \\
    \geq\; & \sum_{j=1}^{r_2} |U_1(\mathbf{x}^{(1)}) \cap R_{j}| (r_2 - 1)
    - 2\sum_{j=2}^{r_2-1} \delta_{|U_1(\mathbf{x}^{(1)}) \cap R_{j}| \neq 0}
    - \delta_{|U_1(\mathbf{x}^{(1)}) \cap R_{1}| \neq 0}
    - \delta_{|U_1(\mathbf{x}^{(1)}) \cap R_{r_2}| \geq 2}
    \nonumber \\
    =\; & r_1 (r_2 - 1) - \Delta, \label{eq: theta_i sum4}
    \end{align}
    where
    $$
    \Delta \triangleq 2\sum_{j=2}^{r_2-1} \delta_{|U_1(\mathbf{x}^{(1)}) \cap R_{j}| \neq 0} + \delta_{|U_1(\mathbf{x}^{(1)}) \cap R_{1}| \neq 0} + \delta_{|U_1(\mathbf{x}^{(1)}) \cap R_{r_2}| \geq 2}.
    $$
    
    Note that for $r_1 \geq r_2 \geq 2$, it always holds that
    $$
    \Delta \leq 2(r_2 - 2) + 1 + 1 = 2(r_2 - 1).
    $$
    Moreover, when $r_1 = r_2 \geq 2$, since $n/2 \in U_1(\mathbf{x}^{(1)}) \cap R_{r_2}$, it holds that
    \begin{align*}
    |U_1(\mathbf{x}^{(1)}) \cap \bigcup_{j=1}^{r_2-1} R_{j}| 
    &= r_1 - |U_1(\mathbf{x}^{(1)}) \cap R_{r_2}| \\
    &= r_2 - |U_1(\mathbf{x}^{(1)}) \cap R_{r_2}|.
    \end{align*}
    This implies that the terms $\delta_{|U_1(\mathbf{x}^{(1)}) \cap R_{j}| \neq 0}$ for $1 \leq j \leq r_2 - 1$, and $\delta_{|U_1(\mathbf{x}^{(1)}) \cap R_{r_2}| \geq 2}$ cannot all take the value $1$ simultaneously. Thus, we conclude that
    \begin{align}\label{eq: Delta}
    \Delta \leq 2(r_2 - 1) - \delta_{r_1 = r_2}.
    \end{align}
    Plugging this upper bound on $\Delta$ into \eqref{eq: theta_i sum4}, we obtain
    \begin{align}\label{eq: theta_i sum5}
    \sum_{i \in U_1(\mathbf{x}^{(1)})} \left(|U_1(D_1(\mathbf{x}^{(2)}, i))| - \delta_{x^{(2)}_{i - 1} \neq x^{(2)}_{i + 1}} - 1\right) + \delta_{r_{\mathrm{last}}(\mathbf{x}^{(2)}) = 1}
    \geq (r_1 - 2)(r_2 - 1) + \delta_{r_1 = r_2}.
    \end{align}
    
    %we have $|U_1(D_1(\mathbf{x}^{(2)}, i))|=1$ and $\delta_{x^{(2)}_{i - 1} \neq x^{(2)}_{i + 1}}=0$ holds for all $i\in [n/2]$. Thus, by $n/2\geq 2$, this leads to
    % \begin{align*}
    %     & \sum_{i\in U_1(\mathbf{x}^{(1)})} \left(|U_1(D_1(\mathbf{x}^{(2)}, i))| - \delta_{x^{(2)}_{i - 1} \neq x^{(2)}_{i + 1}} - 1\right) + \delta_{r_{\mathrm{last}}(\mathbf{x}^{(2)}) = 1} =0.
    % \end{align*}
    % When $r_1 = r_2 > 1$, we have
    % \begin{align*}
    %     & \sum_{i\in U_1(\mathbf{x}^{(1)})} \left(|U_1(D_1(\mathbf{x}^{(2)}, i))| - \delta_{x^{(2)}_{i - 1} \neq x^{(2)}_{i + 1}} - 1\right) + \delta_{r_{\mathrm{last}}(\mathbf{x}^{(2)}) = 1} \nonumber \\
    %     \geq&(r_2 - 2)(r_2 - 3) + (r_2 - 2) + (r_2 - 1)\\
    %     =&(r_1 - 2)(r_2 - 1) + 1.
    % \end{align*} 
    % When $r_1 = r_2 = 1$, we have
    % \begin{align*}
    %     &\sum_{i\in U_1(\mathbf{x^1})} (|U_1(D_1^1(\mathbf{x^2}, i))| - \delta_{\mathbf{x^2}_{i - 1} \neq \mathbf{x^2}_{i + 1}} - 1) + \delta_{r_{\mathrm{l}}(\mathbf{x^2}) = 1}\\
    %     \geq&r_2 - 1 = 0.
    % \end{align*} 

    Finally, combining \eqref{eq: lb of D_2(x)}, \eqref{eq: U_1(x_{n-r_L})}, \eqref{eq: r_2=1}, and \eqref{eq: theta_i sum4}, we conclude that the right-hand side of \eqref{eq: D_2^2(x)} is at least
    \begin{align*}
    |D_2(V_{r_1})| + |D_2(V_{r_2})| - r_2 + (r_1 - 2)(r_2 - 1) + \delta_{r_1 = r_2} - \delta_{r_1 = r_2 = 1}.
    \end{align*}
    Furthermore, one can easily verify that for $r_1 \geq r_2$, the sequence $\mathbf{x}$ with array representation
    $\begin{bmatrix}
    V_{r_1};
    V_{r_2}
    \end{bmatrix}$
    has deletion ball size $|D_2^2(\mathbf{x})|$ achieving this lower bound.

    This completes the proof.
\end{IEEEproof}

\begin{example}\label{example: D_2^2(x)}
    Let $n = 8$, $r_1 = 4$, and $r_2 = 2$. Then we have $V_4 = 0101$ and $V_2 = 0111$. One can easily verify that $|D_2(V_4)|=4$, $|D_2(V_2)|=2$, and the deletion ball $D_2^{2}(\mathbf{x})$ with $\mathbf{x}=00110111$ (whose array representation is $[V_4; V_2]$) equals
    $$
    D_2^{2}(\mathbf{x})=\{0111,0011,1111,1101,0101,0001\},
    $$
    and thus by Lemma \ref{lem: LL}, $|D_2^{2}(\mathbf{x})|=6$
    achieves the lower bound in Theorem \ref{thm: lower_bound_for_2}. 
    %i.e., for any $\mathbf{x} = \mathbf{x^1}||\mathbf{x^2}\in\Sigma_2^n$ with $\mathbf{x^1}$ contains $r_1$ runs and $\mathbf{x^2}$ contains $r_2$ runs, we have $|D_2^2(\mathbf{x})| \geq |D_2^2(V_1||V_2)| = 6$.
    %When $n = 10$, $r_1 = 3$, and $r_2 = 5$, let $V_3' = 00010$ be the reverse of $V_3$ and $V_5' = 01010$ be the reverse of $V_5$. It can be verified that $V_3'||V_5'$ achieves the lower bound, i.e., for any $\mathbf{x} = \mathbf{x^1}||\mathbf{x^2}\in\Sigma_2^n$ with $\mathbf{x^1}$ contains $r_1$ runs and $\mathbf{x^2}$ contains $r_2$ runs, we have $|D_2^2(\mathbf{x})| \geq |D_2^2(V_3'||V_5')| = 13$.
\end{example}

As a direct corollary of Lemma \ref{lem: LL} and Theorem \ref{thm: lower_bound_for_2}, we obtain the following explicit lower bound on $|D_2^2(\mathbf{x})|$ for any $\mathbf{x} \in \Sigma_2^n$ with $2 \mid n$.

\begin{corollary}\label{coro: explicit lb on D_2^2}
    For any even positive integer $n$ and integers $1 \leq r_1, r_2 \leq n/2$, define
    \begin{align*}
    U(n, r_1, r_2) &\triangleq u\Big(\frac{n}{2}, r_1\Big) + u\Big(\frac{n}{2}, r_2\Big) - \min\{r_1, r_2\} \\
    & \quad + (\max\{r_1, r_2\} - 2)(\min\{r_1, r_2\} - 1) \\
    & \quad + \delta_{r_1 = r_2} - \delta_{r_1 = r_2 = 1}.
    \end{align*}
    Then, for any sequence $\mathbf{x} \in\Sigma_q^n$ with array representation 
    $\begin{bmatrix}
    \mathbf{x}^{(1)};
    \mathbf{x}^{(2)}
    \end{bmatrix}$
    such that the numbers of runs of $\mathbf{x}^{(1)}$ and $\mathbf{x}^{(2)}$ are $r_1$ and $r_2$, respectively, it holds that
    $$
    |D_2^2(\mathbf{x})| \geq U(n, r_1, r_2).
    $$
\end{corollary}

In the same spirit as the general upper bound obtained via a linear programming approach, we use Theorem~\ref{thm2: upper bound on M_q(n,(t,b))} and Corollary~\ref{coro: explicit lb on D_2^2} to derive the following improved upper bound on $M_q(n,(2,2))$ for even $n$.

\begin{theorem}\label{thm: ub_on_M_q(n,(2,2))}
    For any $n \geq 6$ and integer $q \geq 2$, it holds that
    \begin{align*}
    M_q(n,(2,2)) \leq \sum_{r_1 = 1}^{\lceil \frac{n}{2}\rceil -2} \sum_{r_2 = 1}^{\lceil \frac{n}{2}\rceil-2} 
    \frac{q^2 (q - 1)^{r_1 + r_2 - 2} 
    \binom{\lceil \frac{n}{2}\rceil-3}{r_1 - 1} 
    \binom{\lceil \frac{n}{2}\rceil-3}{r_2 - 1}}{U(n, r_1, r_2)}.
    \end{align*}
\end{theorem}

\begin{IEEEproof}
    First, we show that it suffices to prove the upper bound for the case when $n$ is even.

    Suppose $n$ is odd and let $\mathcal{C} \subseteq \Sigma_q^{n}$ be a $(2,2)$-burst-deletion-correcting code. Define
    $$
    \mathcal{C}' \triangleq \{(c_1,c_2,\ldots,c_n,0) \in \Sigma_q^{n+1} :~ (c_1,c_2,\ldots,c_n)\in \mathcal{C}\}.
    $$
    We claim that $\mathcal{C}'$ is also a $(2,2)$-burst-deletion-correcting code, which implies
    $$
    M_q(n,(2,2)) \leq M_q(n+1,(2,2)).
    $$
    Hence, it suffices to consider even $n$. For simplicity, write $\mathbf{c}0$ for $(c_1,\ldots,c_n,0)$.
    
    Suppose, for contradiction, that there exist $\mathbf{x}0,\mathbf{y}0 \in \mathcal{C}'$ and $\mathbf{z} \in \Sigma_q^{n-3}$ such that
    \begin{equation}\label{eq1: thm ub on M_q(n,(2,2))}
    \mathbf{z} \in D_{2}^{2}(\mathbf{x}0)\cap D_{2}^{2}(\mathbf{y}0).
    \end{equation}
    We obtain a contradiction by considering the following two cases.
    \begin{itemize}
        \item \emph{Case 1:} $z_{n-3}=0$. Since both $\mathbf{x}0$ and $\mathbf{y}0$ end with $0$, we have $\mathbf{z}\in D_{2}^{2}(\mathbf{x})\cap D_{2}^{2}(\mathbf{y})$, which contradicts $\mathbf{x},\mathbf{y}\in \mathcal{C}$.
        \item  \emph{Case 2:} $z_{n-3}\neq 0$. Then one of the deletions must occur at position $n$ for both $\mathbf{x}0$ and $\mathbf{y}0$, and hence 
        $$
        \mathbf{z}\in D_{1}^{2}(\mathbf{x}_{[n-1]})\cap D_{1}^{2}(\mathbf{y}_{[n-1]}).
        $$
        From $\mathbf{z}\in D_{1}^{2}(\mathbf{x}_{[n-1]})$, we consider two subcases: if $z_{n-3}=x_{n-1}$, then $\mathbf{z}_{[n-4]}=\mathrm{Del}_2(\mathbf{x},(i,1),(n-1,1))$ for some $1\le i\le n-3$; otherwise, $\mathbf{z}_{[n-4]}=\mathrm{Del}_2(\mathbf{x},(n-3,2)).$ In both subcases, we obtain $\mathbf{z}_{[n-4]}\in D_2^2(\mathbf{x})$. Similarly, $\mathbf{z}_{[n-4]}\in D_2^2(\mathbf{y})$, again yielding a contradiction.
    \end{itemize}
    %Thus, $\mathcal{C}'$ is $(2,2)$-burst-deletion correcting, and the reduction is proved.

    Next, we complete the proof by establishing the upper bound on $M_q(n,(2,2))$ for even $n$.

    Note that for every $r \in [n/2-2]$, the number of sequences $\mathbf{z} \in \Sigma_q^{n/2-2}$ with $r$ runs is exactly
    $$
    q(q-1)^{r-1} \binom{n/2-3}{r-1}.
    $$
    Thus, for every pair $(r_1, r_2) \in [n/2-2] \times [n/2-2]$, the number of sequences $\mathbf{y} \in \Sigma_q^{n-4}$ with array representation $[\mathbf{y}^{(1)};\mathbf{y}^{(2)}]$ such that $\mathbf{y}^{(1)}$ and $\mathbf{y}^{(2)}$ have $r_1$ and $r_2$ runs, respectively, is equal to
    $$
    q^2 (q - 1)^{r_1 + r_2 - 2} 
    \binom{n/2-3}{r_1 - 1} 
    \binom{n/2-3}{r_2 - 1}.
    $$
    The result then follows directly from Theorem~\ref{thm2: upper bound on M_q(n,(t,b))} and Corollary~\ref{coro: explicit lb on D_2^2}.
\end{IEEEproof}

Figure \ref{fig: comparison of upper bounds 2} compares the non-asymptotic upper bounds on $M_q(n,(t,b))$ in Theorem~\ref{thm: non asymptotic upper bound 1}, Theorem~\ref{thm: upper_bound_levenshtein}, Theorem~\ref{thm: ub_on_M_q(n,(2,2))}, and the upper bound in \cite{LSYG26} for $q = 3$, $t = 2$, and $b = 2$. The asymptotic behavior of the upper bound given in Theorem~\ref{thm: ub_on_M_q(n,(2,2))} is characterized in the following corollary, whose proof is deferred to Appendix~\ref{appen.a} for interested readers.

\begin{figure}[h!]
    \centering
    \includegraphics[width=0.5\textwidth]{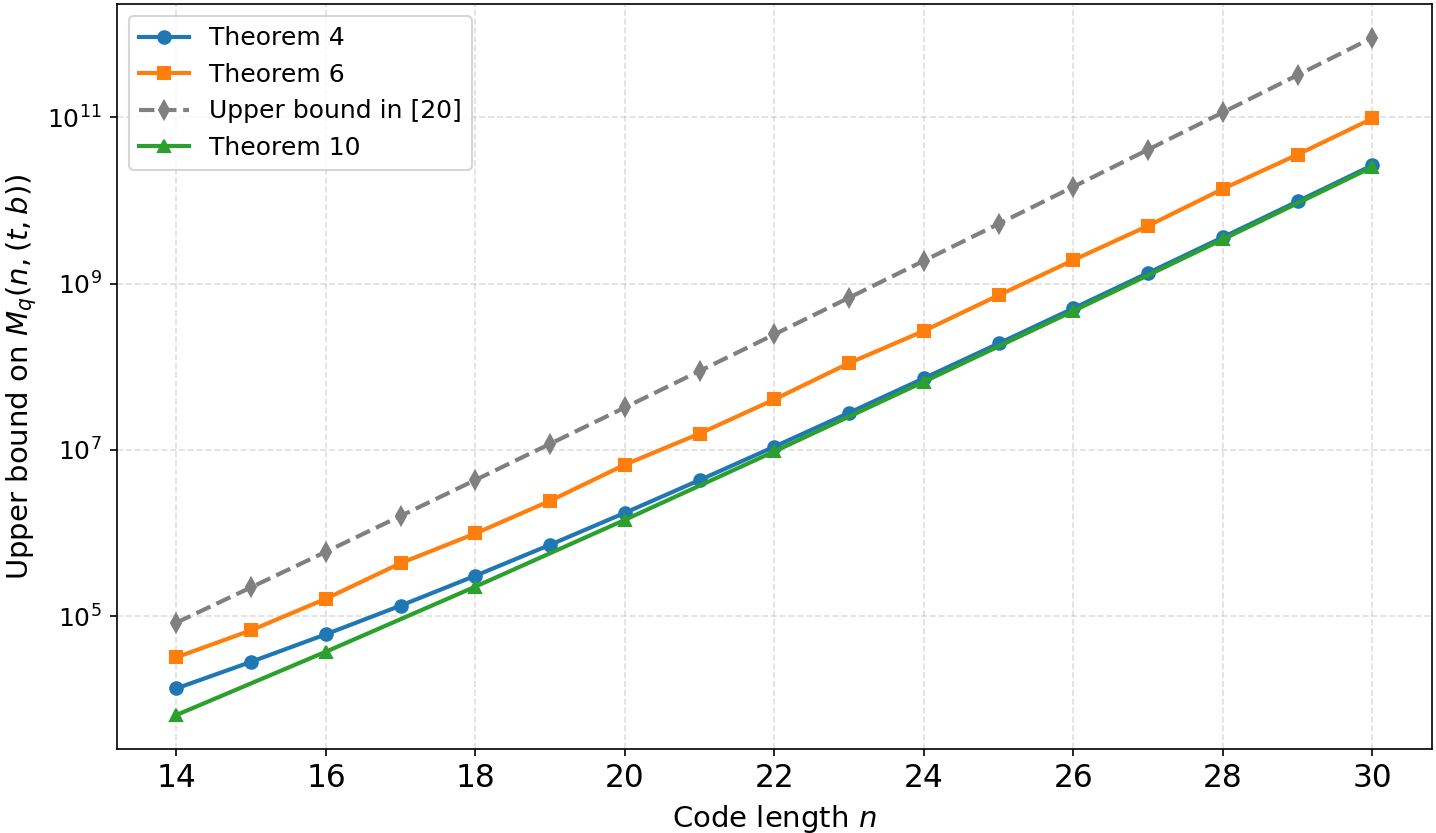}
    \caption{Comparison of the upper bounds on $M_q(n,(t,b))$ in Theorem~\ref{thm: non asymptotic upper bound 1}, Theorem~\ref{thm: upper_bound_levenshtein}, Theorem~\ref{thm: ub_on_M_q(n,(2,2))}, and the upper bound in \cite{LSYG26} for $q = 3$, $t = 2$, $b = 2$.}\label{fig: comparison of upper bounds 2}
\end{figure}

\begin{corollary}\label{coro: asy up for t=b=2}
    For any integer $n$ and $q \geq 2$, it holds that
    \begin{align*}
    M_q(n,(2,2)) \leq \frac{2q^{n-4}}{(q-1)^2n^2}(1+o(1)),
    \end{align*}
    as $n\rightarrow \infty$.
\end{corollary}

\begin{remark}\label{rmk}
    When taking $t=b=2$ in Theorem~\ref{thm1: upper bound on M_q(n,(t,b))} and in the asymptotic version of Theorem~\ref{thm: upper_bound_levenshtein} (see Part~2 of Remark~\ref{rmk: comparison with LP bound}), both upper bounds yield
    $$
    M_q(n,(2,2)) \leq \frac{2q^{n-2}}{(q-1)^2n^2}(1+o(1)).
    $$
    Therefore, by Corollary~\ref{coro: asy up for t=b=2}, our upper bound on $M_q(n,(2,2))$ improves upon those in Theorem~\ref{thm1: upper bound on M_q(n,(t,b))} and Theorem~\ref{thm: upper_bound_levenshtein} by a factor of $q^{-2}$.
\end{remark}

\section{Conclusion and further research}\label{sec: conclusion}

In this paper, we investigate the fundamental limits of codes designed to correct multiple $b$-burst deletions. By analyzing the structural properties of $(t, b)$-burst deletion balls, we establish a monotonicity property of the ball size and introduce the notion of maximal compact representation to uniquely represent each sequence within a $(t, b)$-burst deletion ball. Based on these properties, we derive two upper bounds on the maximum code size $M_q(n, (t, b))$ using linear programming and sphere-packing techniques. Our results improve upon existing bounds in general parameter regimes and recover known results for the special cases with $t = 1$ and $b = 1$. Moreover, we establish combinatorial upper and lower bounds that characterize the asymptotic behavior of $M_q(n, (t, b))$ when $q$ is sufficiently large and $n, t, b$ are constants. Additionally, we obtain a tighter bound for the specific case $t = b = 2$ through a refined analysis of the deletion ball structure.

In the following, we outline several directions for future research.
\begin{itemize}
    \item First, while this work establishes theoretical limits, the construction of $(t, b)$-burst deletion-correcting codes with redundancy approaching these bounds, along with efficient decoding algorithms, remains an important open problem. Currently, only the case $b=1$ is known to achieve the upper bound up to a constant factor.
    \item Second, the gap between the upper and lower bounds could be further reduced, particularly for small alphabet sizes, potentially through more sophisticated linear programming approaches.
    \item Third, the techniques developed for the case $t = b = 2$ could be generalized to other values of $t$ and $b$ to obtain tighter universal bounds. 
    \item Finally, given the complexity of errors in practical applications such as DNA-based storage, it would be of interest to extend this study to channels subject to other types of burst deletions and insertions.
\end{itemize}

%\section{Acknowledgement}

\section*{Appendix~A}\label{appen.a}

\begin{IEEEproof}[Proof of Corollary \ref{coro: asy up for t=b=2}]
By Lemma~\ref{lem: LL}, we know that 
$$
u(n,r)=d(r,2)+1=\frac{r^2-3r}{2}+3
$$
for $3\leq r\leq n-2$. Moreover, note that
$$
(\max\{r_1, r_2\} - 2)(\min\{r_1, r_2\} - 1)=r_1r_2+O(r_1)+O(r_2).
$$
Thus, we have
\begin{align}
    U(n, r_1, r_2) & = \frac{r_1^2}{2}+\frac{r_2^2}{2}+r_1r_2+O(r_1)+O(r_2) \nonumber \\
    & = \frac{(r_1+r_2)^2}{2}(1+o(1)), \label{eq1: asy b=t=2}
\end{align}
as $r_1+r_2\rightarrow \infty$. This implies that 
\begin{align}
    \sum_{r_1 = 1}^{n/2-2} \sum_{r_2 = 1}^{n/2-2}\frac{(q - 1)^{r_1 + r_2 - 2} 
    \binom{n/2-3}{r_1 - 1} 
    \binom{n/2-3}{r_2 - 1}}{U(n, r_1, r_2)} \leq 2\sum_{r_1 = 1}^{n/2-2} \sum_{r_2 = 1}^{n/2-2} \frac{(q - 1)^{r_1 + r_2 - 2} 
    \binom{n/2-3}{r_1 - 1} 
    \binom{n/2-3}{r_2 - 1}}{(r_1+r_2)^2}(1+o(1)) \label{eq2: asy b=t=2}
\end{align}
as $n\rightarrow \infty$. Denote 
\begin{align}
    S_n \triangleq \sum_{s_1 = 0}^{n/2-3} \sum_{s_2 = 0}^{n/2-3} \frac{(q - 1)^{s_1 + s_2} 
    \binom{n/2-3}{s_1} 
    \binom{n/2-3}{s_2}}{(s_1+s_2+2)^2}. \nonumber
\end{align}
Then, by \eqref{eq2: asy b=t=2}, it suffices to show that
$$
S_n=\frac{q^{n-4}}{(q-1)^2n^2}(1+o(1)).
$$

Note that $\sum_{i=0}^{n/2-3}(q-1)^{i}{n/2-3 \choose i}=q^{n/2-3}$. Thus, we can write
\begin{align}
    S_n & = q^{n-6} \sum_{s_1 = 0}^{n/2-3} \sum_{s_2 = 0}^{n/2-3} \frac{q^{-(n-6)+(s_1+s_2)}(1 - \frac{1}{q})^{s_1 + s_2} 
    \binom{n/2-3}{s_1} 
    \binom{n/2-3}{s_2}}{(s_1+s_2+2)^2} \nonumber \\
    & = q^{n-6}\mathbb{E}[(X_1+X_2+2)^{-2}] \nonumber \\
    & = q^{n-6}\mathbb{E}[(X+2)^{-2}], \label{eq3: asy b=t=2}
\end{align}
where $X_1,X_2\sim \mathrm{B}(\frac{n}{2}-3,1-\frac{1}{q})$ are independent random variables, and $X\sim \mathrm{B}(n-6,1-\frac{1}{q})$ follows from the fact that the sum of two independent binomial variables with the same probability $p$ is again a binomial variable; see \cite[Section 11.1]{DKLM05} for more details. 

Let $f(x)=(x+2)^{-2}$. As a function of $x$, the Taylor expansion of $f$ at $\mu\triangleq \mathbb{E}[X]$ is
\begin{align}
    f(x) & = \sum_{k=0}^{\infty}\frac{f^{(k)}(\mu)}{k!}(x-\mu)^{k} \nonumber \\
    & = f(\mu)-\frac{2(x-\mu)}{(\mu+2)^3}+\frac{3(x-\mu)^2}{(\mu+2)^{4}}+\sum_{k=3}^{\infty}\frac{f^{(k)}(\mu)}{k!}(x-\mu)^{k} \nonumber \\
    & = f(\mu)-\left(\frac{2(x-\mu)}{(1-\frac{1}{q})^{3}n^3}-\frac{3(x-\mu)^2}{(1-\frac{1}{q})^{4}n^{4}}\right)(1+O(\frac{1}{n}))+O\left(\frac{(x-\mu)^3}{n^5}\right) \label{eq4: asy b=t=2}
\end{align}
as $n\rightarrow \infty$, where the last equality follows from $\mu=(n-6)(1-\frac{1}{q})=n(1-\frac{1}{q})(1-\frac{6}{n})$. Then, by the linearity of expectation, we have
\begin{align}
    \mathbb{E}[f(X)] & = \mathbb{E}[f(\mu)]+\left(\frac{2\mathbb{E}[X-\mu]}{(1-\frac{1}{q})^{3}n^3}+\frac{3\mathbb{E}[(X-\mu)^2]}{(1-\frac{1}{q})^{4}n^4}\right)(1+O(\frac{1}{n}))+O\left(\frac{\mathbb{E}[(X-\mu)^3]}{n^5}\right) \nonumber \\
    & = \left(\frac{q^2}{(q-1)^2n^2}+\frac{3q^2}{(q-1)^3n^3}\right)(1+O(\frac{1}{n}))+O\left(\frac{\mathbb{E}[(X-\mu)^3]}{n^5}\right), \label{eq5: asy b=t=2}
\end{align}
where the last equality follows from $f(\mu)=(\mu+2)^{-2}=\frac{q^2}{(q-1)^2n^2}(1+O(\frac{1}{n}))$, $\mathbb{E}[X-\mu]=0$, and $\mathbb{E}[(X-\mu)^2]=\mathrm{Var}(X)=\frac{(q-1)(n-6)}{q^2}$. Moreover, by the well-known Hoeffding inequality, we have
$$
\mathrm{Pr}(|X- \mu|\geq n^{2/3})\leq 2e^{-2n^{1/3}}\leq \frac{1}{n},
$$
as $n\rightarrow \infty$. This yields that
\begin{align}
    \mathbb{E}[(X-\mu)^3] & \leq n^2+ \mathrm{Pr}(|X- \mu|\geq n^{2/3})\cdot n^3  =2n^2. \label{eq6: asy b=t=2}
\end{align}
Therefore, by substituting \eqref{eq6: asy b=t=2} into \eqref{eq5: asy b=t=2}, we have
\begin{align}
    \mathbb{E}[f(X)] & = \frac{q^2}{(q-1)^2n^2}+O(\frac{1}{n^3}). \nonumber
\end{align}
This leads to $S_n=\frac{q^{n-4}}{(q-1)^2n^2}(1+o(1))$ and completes the proof.
\end{IEEEproof}

\balance
\bibliographystyle{IEEEtran}
\bibliography{reference}

\end{document}